\documentclass[11pt]{article}

\usepackage[margin=1in]{geometry}
\usepackage{amsmath,amsfonts,amssymb,amsthm,mathtools}
\usepackage{bm}
\usepackage[round,authoryear]{natbib}
\usepackage{microtype}
\usepackage{booktabs}
\usepackage{enumitem}
\usepackage{graphicx}
\usepackage{xcolor}
\usepackage{float}
\usepackage{placeins}
\usepackage{url}
\usepackage[hidelinks]{hyperref}
\usepackage[nameinlink,capitalise,noabbrev]{cleveref}
\allowdisplaybreaks

\newtheorem{theorem}{Theorem}
\newtheorem{proposition}[theorem]{Proposition}
\newtheorem{lemma}[theorem]{Lemma}
\newtheorem{corollary}[theorem]{Corollary}

\crefname{theorem}{Theorem}{Theorems}
\crefname{proposition}{Proposition}{Propositions}
\crefname{lemma}{Lemma}{Lemmas}
\crefname{corollary}{Corollary}{Corollaries}
\crefname{definition}{Definition}{Definitions}
\crefname{assumption}{Assumption}{Assumptions}
\crefname{remark}{Remark}{Remarks}

\DeclareMathOperator{\Tr}{Tr}
\newcommand{\supp}{\operatorname{supp}}
\newcommand{\id}{\mathbb{I}}

\newcommand{\cS}{\mathsf{S}}
\newcommand{\cE}{\mathcal{E}}
\newcommand{\cD}{\mathcal{D}}
\newcommand{\cN}{\mathcal{N}}
\newcommand{\cR}{\mathcal{R}}
\newcommand{\cP}{\mathcal{P}}
\newcommand{\cT}{\mathcal{T}}
\newcommand{\cL}{\mathcal{L}}
\newcommand{\dd}{\mathrm{d}}
\newcommand{\Fr}{F_{\mathrm r}}
\newcommand{\dtr}{d_{\mathrm{tr}}}
\newcommand{\Wtr}{W_{1,\mathrm{tr}}}
\newcommand{\half}{\tfrac{1}{2}}
\newcommand{\ket}[1]{\lvert #1\rangle}
\newcommand{\ketbra}[1]{\lvert #1\rangle\!\langle #1\rvert}
\newcommand{\norm}[1]{\left\lVert #1\right\rVert}
\newcommand{\abs}[1]{\left\lvert #1\right\rvert}
\newcommand{\pos}[1]{\left[#1\right]_{+}}
\newcommand{\idmap}{\operatorname{id}}
\DeclareMathOperator{\Law}{Law}

\title{Information-Calibrated Quantum Diffusion: Aligning Forward Noise with Reverse Recoverability}

\author{
Qipeng Qian$^{1,2}$ \quad Jialiang Jin$^{3}$ \quad Shihao Wang$^{4}$ \quad Yuntao Qian$^{1}$\\[4pt]
$^{1}$College of Artificial Intelligence, Zhejiang University\\
$^{2}$SUPCON Technology, Hangzhou, China\\
$^{3}$Department of Computer Science, Johns Hopkins University\\
$^{4}$Department of Mathematics, University of Wisconsin-Madison
}
\date{}

\begin{document}
\maketitle

\begin{abstract}
Quantum diffusion models typically parameterize forward corruption by raw
channel strength, even though this parameter does not directly quantify how
much ensemble information is erased or how difficult the corresponding
reverse problem is. We introduce the classical--quantum information
decrement $\Delta_t=I(X{:}Q_{t-1})-I(X{:}Q_t)$ as an intrinsic diffusion
coordinate that links forward noise allocation to reverse recoverability.
Along depolarization, equalizing $\Delta_t$ yields the unique minimax
discretization of the forward information loss, while universal
recoverability gives the same quantity an operational interpretation as a
physically attainable local recovery budget. We further show that such
local calibration is not sufficient for stochastic generation: models can
satisfy the same recovery criterion while producing substantially different
state distributions. This motivates a stochastic learner that combines
information-calibrated recovery constraints with distribution matching. We
establish finite-sample calibration and compositional trace-Wasserstein
control for the resulting learner. Controlled quantum experiments validate
the predicted information--recovery alignment, show that the recovery
constraints improve local inversion, and confirm the complementary role of
distribution matching in endpoint generation. The resulting framework also
achieves stronger endpoint trace-Wasserstein performance than an official
QuDDPM implementation with fewer trainable parameters. Overall, our work
provides a unified information-theoretic principle for designing forward
schedules, calibrating reverse steps, and separating physical recovery from
generative coverage in quantum diffusion.
\end{abstract}

\section{Introduction}

Diffusion models generate complex data by learning to reverse a prescribed
forward corruption process. In classical diffusion, the forward path is
typically parameterized by a noise-strength schedule, with linear, cosine,
and other hand-designed schedules determining how corruption is allocated
across timesteps \citep{sohl2015deep,ho2020ddpm,nichol2021improved,karras2022design}.
Quantum diffusion models adopt the same basic principle with quantum
channels, including depolarizing processes, and learn reverse quantum
dynamics or ensemble-level generative objectives
\citep{zhang2024generative,zhang2024qgdm,kwun2024mixed,xu2026ssdm}.
These choices define the sequence of inverse problems presented to the
reverse model. However, the raw channel parameter does not itself quantify
how much information about the data ensemble is erased by a step, nor does
it provide a physical scale for how difficult that step should be to
reverse.

This gap is particularly relevant for quantum-state generation. The target
is generally an ensemble of states rather than a single average density
operator, and two ensembles can encode different probability laws even when
their averages coincide \citep{oreshkov2009ensembles}. Existing quantum
diffusion models already use state- or ensemble-level objectives to learn
such distributions \citep{zhang2024generative,kwun2024mixed}, but the
forward schedule and the reverse objective are usually calibrated in
different coordinates. This leaves a basic question unresolved: can the
forward process be parameterized by a quantity that has an intrinsic
information meaning and, at the same time, determines a physically
attainable scale for the corresponding reverse problem?

We answer this question using the classical--quantum information retained
about the ensemble label. For a labeled ensemble, we define the one-step
information decrement
\begin{equation}
\Delta_t
:=
I(X{:}Q_{t-1})-I(X{:}Q_t).
\label{eq:intro-delta}
\end{equation}
Along depolarization, the resulting Holevo-information curve defines an
intrinsic diffusion clock. Equalizing $\Delta_t$ yields the unique grid
that minimizes the largest one-step information loss among all monotone
grids with fixed endpoints. More importantly, universal recoverability
turns the same decrement into an operational reverse scale: a single
label-independent CPTP recovery channel can attain an expected
log-fidelity error no larger than $\Delta_t$
\citep{buscemi2016approximate,junge2018universal}. Continuity arguments
provide the converse direction, linking vanishing information loss to
vanishing intrinsic common-channel recovery error. Thus the same scalar
coordinate connects forward noise allocation and local reverse
recoverability.

Local recoverability, however, is not the whole generative problem.
Quantum diffusion models ultimately need to reproduce a probability law
over states, and prior quantum diffusion work therefore relies on
distribution-level objectives in addition to state-wise similarity
\citep{zhang2024generative,kwun2024mixed}. We show that this distinction
remains essential even after the local inverse has been physically
calibrated. In a fixed noncommuting two-qubit construction, two stochastic
CPTP reverse models can have identical local-fidelity statistics and both
satisfy the same information-calibrated recovery budget while producing
substantially different output distributions. This establishes that
stepwise physical recovery and generative coverage are complementary,
non-substitutable requirements.

These results lead naturally to a stochastic learner that minimizes
distributional mismatch subject to timestep-wise recovery constraints
whose scales are set by $\Delta_t$. We establish finite-sample calibration
of the learned constraints and a compositional trace-Wasserstein endpoint
certificate, thereby connecting the forward information clock, learned
local inversion, and final state-distribution error. Controlled two- and
four-qubit experiments validate the predicted information--recovery
alignment, show that enforcing the budgets improves local inversion, and
empirically reproduce the recoverability--coverage separation. On the
four-qubit TFIM task, increasing reverse-model capacity substantially
improves endpoint generation, and the resulting learner achieves lower
endpoint trace-Wasserstein error than the official QuDDPM implementation
while using fewer trainable parameters. Independently seeded schedule
experiments further support the equal-information allocation over standard
cosine scheduling.

Our contributions are:
\begin{itemize}[leftmargin=*]
\item \textbf{An operational information coordinate for quantum diffusion.}
We introduce the cq-information decrement $\Delta_t$ as a diffusion
coordinate that simultaneously measures forward ensemble-information loss
and calibrates a physically attainable common-channel reverse-recovery
budget. Along depolarization, equal-information allocation is the unique
minimax discretization of this coordinate.

\item \textbf{A recoverability--coverage separation for stochastic quantum
generation.}
We prove in a fixed noncommuting two-qubit system that satisfying the
information-calibrated local recovery criterion does not determine the
generated state distribution. This identifies distribution matching as a
separate requirement rather than a surrogate for local physical recovery.

\item \textbf{A theorem-calibrated stochastic learner with local-to-global
guarantees and mechanism-level validation.}
We combine distribution matching with timestep-wise recovery constraints,
establish finite-sample calibration and compositional trace-Wasserstein
control, and validate the resulting decomposition through controlled
ablations, capacity extension, schedule comparisons, and an external
QuDDPM endpoint benchmark.
\end{itemize}

\section{Related work}
\label{sec:related}
Classical diffusion uses a wide range of noise schedules and denoising objectives \citep{sohl2015deep,ho2020ddpm,song2021ddim,song2021scorebased,kingma2021variational,nichol2021improved,karras2022design}, including recent entropy- or information-guided allocation schemes \citep{stancevic2025entropic,raya2026infonoise,ambrogioni2026atoms,foresti2026schedules}. Quantum diffusion has developed trainable reverse circuits, mixed-state depolarization, temporal conditioning, measurement-conditioned dynamics, stochastic reverse processes, and recovery-based constructions \citep{zhang2024generative,zhang2024qgdm,kwun2024mixed,liu2026measurement,xu2026ssdm,bompais2026reverse}. Mixed-state quantum diffusion supplies depolarizing forward paths and ensemble-level similarity objectives \citep{kwun2024mixed}. We instead make ensemble information the diffusion coordinate, derive the operational reverse scale induced by that coordinate, and identify the additional distributional control required for stochastic generation.

Our reverse analysis builds on Holevo information, relative-entropy data processing, Petz/universal recovery, and Holevo continuity \citep{holevo1973bounds,lindblad1975maps,petz1986sufficient,fawzi2015approximate,junge2018universal,buscemi2016approximate,shirokov2017continuity}. Universal-recovery inequalities provide the common-channel primitive. We turn it into the expected-log budget induced by each information-calibrated forward step, pair it with an independent geometry-sensitive converse, and characterize the stochastic generative freedom left after local calibration. Randomization and deficiency give a complementary decision-theoretic view of common-channel simulation \citep{jencova2016comparison}; our pairwise LP turns forward-erased distinguishability into a direct computable recovery witness.

Unrestricted common-channel recoverability therefore serves as a \emph{reference scale} for the inverse problems created by the information-calibrated forward process, rather than a reverse dynamics to imitate. The same $\Delta_t$ appears as forward scheduling coordinate, physically attainable local budget, and learning constraint. The geometry-aware converse retains which pairwise distinctions were erased, information that is invisible to a scalar continuity modulus even when the total Holevo decrement is fixed; the fixed-dimensional separation then identifies the distributional degree of freedom that neither local quantity controls.

Our external benchmark uses the official QuDDPM implementation \citep{zhang2024generative} on identical primary TFIM datasets and the same held-out evaluator. QuDDPM retains its native forward/reverse construction. Our information-calibrated shared learner achieves lower endpoint trace-Wasserstein error with fewer trainable parameters; this claim is intentionally specific to $\Wtr$, the distributional geometry propagated by our local-to-global theorem, rather than a claim of uniform dominance over every endpoint diagnostic.

Quantum-ensemble distances explain why independent state pairing does not define ensemble distinguishability \citep{oreshkov2009ensembles}. Our separation sharpens this observation quantitatively and ties it to physical recovery calibration: a perfectly covering generator and a collapsed generator can share the complete marginal fidelity law and the same budget-feasible local log-risk on a fixed noncommuting two-qubit ensemble. This connects to reconstruction--distribution tensions in classical inverse problems \citep{blau2018perception}, with the local criterion here fixed by quantum recoverability.

\section{Information-calibrated forward diffusion and reverse recoverability}
\label{sec:setting}
We first identify the intrinsic scale of the local inverse problem. Let $\cE=\{p_x,\rho_x\}$ be a nontrivial finite ensemble on a $d\ge2$ dimensional Hilbert space, with $p_x>0$. All logarithms are natural. We use root fidelity and trace distance,
\begin{equation}
  \Fr(\rho,\sigma):=\norm{\sqrt{\rho}\sqrt{\sigma}}_1,
  \qquad
  \dtr(\rho,\sigma):=\half\norm{\rho-\sigma}_1,
  \label{eq:metric-conventions-main}
\end{equation}
and the trace-cost Wasserstein distance
\begin{equation}
  \Wtr(\mu,\nu)
  :=\inf_{\pi\in\Pi(\mu,\nu)}
  \int \dtr(\rho,\sigma)\,\pi(\dd\rho,\dd\sigma).
  \label{eq:wasserstein-definition-main}
\end{equation}
The cq mutual information of the ensemble is the Holevo quantity
\begin{equation}
\chi(\cE)=I(X{:}Q)=S(\bar\rho)-\sum_xp_xS(\rho_x),
\qquad \bar\rho=\sum_xp_x\rho_x.
\label{eq:holevo-identity}
\end{equation}
Keeping the label $X$ is essential for generation: two ensembles can have the same average density operator while defining different probability laws over states.

Along the cumulative depolarizing path 
\begin{equation}
\cD_\lambda(A)
=
\lambda A
+
(1-\lambda)\Tr(A)\frac{\id}{d},
\qquad
\chi(\lambda)
=
\chi\!\left(\{p_x,\cD_\lambda(\rho_x)\}\right),
\label{eq:chi-curve-main}
\end{equation}
we fix the endpoints 
\begin{equation}
\lambda_0=1,
\qquad
\lambda_T<1,
\end{equation}
and define the total information loss along the forward path as 
\begin{equation}
\Delta_{\rm tot}
:=
\chi(\lambda_0)-\chi(\lambda_T)
=
\chi(1)-\chi(\lambda_T).
\label{eq:total-information-loss-main}
\end{equation}

\begin{theorem}[Depolarizing cq-information clock]
\label{thm:forward-clock}
For a nontrivial ensemble, the Holevo-information curve
$\chi(\lambda)$ is continuous and strictly increasing on $[0,1]$.
Therefore, for any fixed number of steps $T$ and endpoints
$\lambda_0=1$ and $\lambda_T<1$, there exists a unique monotone grid
\begin{equation}
1=\lambda_0>\lambda_1>\cdots>\lambda_T
\end{equation}
such that the Holevo information decreases by the same amount at every
forward step:
\begin{equation}
\chi(\lambda_t)
=
\left(1-\frac{t}{T}\right)\chi(\lambda_0)
+
\frac{t}{T}\chi(\lambda_T),
\qquad
t=0,\ldots,T.
\label{eq:equal-information-grid-main}
\end{equation}
Equivalently,
\begin{equation}
\Delta_t
:=
\chi(\lambda_{t-1})-\chi(\lambda_t)
=
\frac{\Delta_{\rm tot}}{T},
\qquad
t=1,\ldots,T.
\label{eq:equal-information-decrement-main}
\end{equation}

Among all monotone grids with the same endpoints and the same number of
steps, this equal-information grid uniquely minimizes the largest
one-step information loss,
\begin{equation}
\max_{1\leq t\leq T}\Delta_t.
\label{eq:minimax-information-grid-main}
\end{equation}
More generally, for every strictly convex function $\varphi$, it
uniquely minimizes
\begin{equation}
\sum_{t=1}^T \varphi(\Delta_t).
\label{eq:convex-information-grid-main}
\end{equation}
\end{theorem}

The theorem defines diffusion time directly in terms of ensemble information rather than the raw depolarizing parameter. In practice, the clock is obtained by evaluating the one-dimensional curve $\chi(\lambda)$ and inverting it at the equally spaced information levels in \cref{eq:equal-information-grid-main}. We write 
\begin{equation}
\rho_{x,t}
:=
\cD_{\lambda_t}(\rho_x),
\qquad
t=0,\ldots,T.
\label{eq:forward-state-main}
\end{equation}

The decrement $\Delta_t$ quantifies the information erased by the $t$th forward step. We next ask whether the same quantity also characterizes how difficult that step is to reverse.
Throughout the following analysis, $\cR$ denotes a deterministic, label-independent CPTP recovery channel.

The equal-information clock determines how the forward path should be
discretized, but it does not yet explain why the decrement
$\Delta_t$ should be relevant to the reverse problem. The key question
is whether the amount of cq information erased by a forward step also
controls how accurately that step can be physically reversed.

For generation, the reverse operation cannot depend on the hidden
label $x$. We therefore require a \emph{common} recovery channel:
a single label-independent CPTP map $\cR_t$ that acts on all members
of the ensemble at step $t$. The following result provides the link
between the forward information clock and such a physical reverse map.

\begin{theorem}[Learner-facing common-channel budget]
\label{thm:recovery-budget}
For every step there exists a common recovery channel $\cR_t$,
depending only on $\bar\rho_{t-1}$ and the known forward step, such that
\begin{equation}
-2\sum_xp_x\log\Fr\left(
\rho_{x,t-1},
\cR_t(\rho_{x,t})
\right)
\leq
\Delta_t.
\label{eq:log-budget-main}
\end{equation}
\end{theorem}

This result gives the forward decrement $\Delta_t$ an operational
reverse interpretation. It is not only the amount of cq information
removed by the $t$th forward step, but also an attainable upper bound
on the expected log-fidelity error of a single physical inverse that
must work across the entire ensemble. Thus the same quantity that
defines the forward schedule also provides a natural scale for local
reverse recovery.

This is the reason for calibrating the diffusion path by cq information rather than directly by the retention parameter $\lambda_t$. A change in $\lambda_t$ only measures displacement along the chosen parameterization of the depolarizing path, whereas $\Delta_t$ measures the ensemble information actually erased by that step. The latter therefore adapts to both the forward channel and the target ensemble before the reverse model is trained.

The preceding result shows that $\Delta_t$ provides an attainable recovery scale. We now ask whether the same information decrement also characterizes the intrinsic difficulty of the reverse step. To place the fidelity-based achievability result and the reverse error in a common metric, we measure recovery quality in trace distance. 

\begin{theorem}[Information-calibrated common-channel recoverability]
\label{thm:two-sided-recoverability}

For the $t$th forward step and any label-independent CPTP recovery channel $\cR$, define the state-wise recovery error and the optimal common-channel recovery error
\begin{equation}
e_x(\cR)
:=
\dtr\!\left(
\rho_{x,t-1},
\cR(\rho_{x,t})
\right),\quad \varepsilon_t^\star
:=
\inf_{\cR\in{\rm CPTP}}
\sum_x p_x\, e_x(\cR).
\label{eq:intrinsic-recovery-error-main}
\end{equation}

Then the information decrement $\Delta_t$ and the intrinsic recovery error $\varepsilon_t^\star$ satisfy
\begin{equation}
\Delta_t
\leq
\Omega_{d,m}\!\left(\varepsilon_t^\star\right),
\label{eq:information-recovery-converse-main}
\end{equation}
where $\Omega_{d,m}$ is the same-prior Holevo-continuity modulus defined in Appendix~\ref{app:recovery}.
Using the Fuchs--van de Graaf inequalities \cite{fuchs1999cryptographic} gives
\begin{equation}
1-\Fr(\rho,\sigma)
\leq
\dtr(\rho,\sigma)
\leq
\sqrt{1-\Fr(\rho,\sigma)^2}.
\label{eq:fvdg-main}
\end{equation}
Combining the upper inequality with the attainable fidelity guarantee
of \cref{thm:recovery-budget} yields
\begin{equation}
\varepsilon_t^\star
\leq
\sqrt{1-e^{-\Delta_t}}.
\label{eq:information-recovery-upper-main}
\end{equation}

Consequently, at fixed Hilbert-space dimension $d$ and ensemble size $m$,
\begin{equation}
\Delta_t\to0
\qquad\Longleftrightarrow\qquad
\varepsilon_t^\star\to0.
\label{eq:information-recovery-equivalence-main}
\end{equation}
\end{theorem}

The theorem gives the information clock a direct reverse interpretation. The forward decrement $\Delta_t$ not only determines how the depolarizing path is discretized, but also calibrates the reverse problem through a physically attainable common-channel recovery scale. A small $\Delta_t$ guarantees the existence of an accurate common inverse, while vanishing optimal recovery error requires the erased Holevo information to vanish as well. Thus the equal-information schedule equalizes an information-level recovery budget rather than claiming that the exact optimal recovery error must be identical at every step.

The information-level characterization can be further strengthened by geometry-sensitive converse bounds based on the pairwise distinguishability erased by the forward channel. These refinements provide additional diagnostics of irreversible ensemble structure but do not modify the information-calibrated schedule, and are therefore deferred to Appendix~\ref{app:pairwise-proof}.

\section{Beyond calibrated recovery: generative coverage}
\label{sec:coverage-separation}

Learning a distribution of quantum states is already a central objective in quantum diffusion models. QuDDPM explicitly formulates the generative task as learning an unknown distribution of quantum states, rather than merely reproducing its average density operator, and trains each stage by comparing the generated and forward ensembles with distribution-level objectives such as MMD or Wasserstein distance \citep{zhang2024generative}. The same ensemble-level viewpoint is extended to mixed states in MSQuDDPM, which uses depolarizing forward noise together with superfidelity-based MMD and Wasserstein objectives for matching mixed-state ensembles \citep{kwun2024mixed}. In particular, these works already recognize that state similarity by itself need not identify a quantum-state distribution; QuDDPM explicitly notes that a fidelity-kernel MMD can fail to distinguish some distributions and uses Wasserstein distance as an alternative \citep{zhang2024generative}.

The question here is therefore not whether distribution matching is needed for quantum diffusion. Rather, the previous section introduced a different object: the information decrement $\Delta_t$ gives a physically attainable scale for \emph{local recovery} across each forward step. Once such a recovery criterion is available, a natural question is whether satisfying it already guarantees the correct generated ensemble, in which case a separate distribution objective would be redundant.

These two requirements concern different aspects of the reverse process. The local recovery criterion asks whether an output state is compatible with the amount of information erased by the corresponding forward step. A distribution-level criterion instead asks whether the stochastic reverse model allocates probability mass correctly across the target ensemble. The purpose of this section is to show that the first requirement does not imply the second, even when the local recovery criterion is satisfied exactly at the scale prescribed by $\Delta_t$.

We first isolate the mechanism with a simple symmetric example. This proposition is only an auxiliary observation; the quantum-specific separation relevant to the information-calibrated recovery budget is given in the theorem that follows.

\begin{proposition}[Coverage blindness of symmetric local scores]
\label{prop:assignment-blindness}
Let the target distribution be uniform over $\{\rho_1,\ldots,\rho_m\}$. Suppose a local score $S$ satisfies 
\begin{equation}
S(\rho_i,\rho_i)=s_{\rm self},
\qquad
S(\rho_i,\rho_j)=s_{\rm off}
\quad
\text{for all } i\neq j.
\label{eq:score-regularity-main}
\end{equation}

Let $X$ and $Z$ be independent and uniform on $\{1,\ldots,m\}$. Consider a correctly covering generator 
\begin{equation}
\widehat\rho^{\rm good}=\rho_Z
\end{equation}
and a collapsed generator
\begin{equation}
\widehat\rho^{\rm col}=\rho_{j_0},
\end{equation}
where $j_0$ is any fixed index.

For both generators, the local score has the same distribution:
\begin{equation}
S(\rho_X,\widehat\rho)
=
\begin{cases}
s_{\rm self}, & \text{with probability }1/m,\\
s_{\rm off}, & \text{with probability }1-1/m.
\end{cases}
\end{equation}
Hence the complete local-score distribution cannot distinguish the correctly covering generator from the collapsed one, even though their output distributions are different.
\end{proposition}

The proposition shows what would be lost if the distribution-level comparison were replaced by a local target--output score alone: even the full distribution of that score can be blind to how generated probability mass is allocated across output modes. 

For the present framework, the remaining question is more specific. Can this blindness persist for the root-fidelity criterion used by our information-calibrated recovery theory, while both models satisfy the same $\Delta$-based physical recovery budget? The next theorem shows that it can, in a fixed genuinely quantum system.

\begin{theorem}[Fixed two-qubit recoverability--coverage separation]
\label{thm:coverage-separation-main}
There exists a one-step complete-depolarization problem on two qubits with $16$ equiprobable pairwise noncommuting target states and two physically valid stochastic CPTP reverse models with the following properties. One model reproduces the target distribution exactly, whereas the other is completely collapsed onto a single target state. Nevertheless,

\begin{enumerate}[label=(\roman*),leftmargin=*]
\item the two models have identical local root-fidelity distributions,
\begin{equation}
\Fr(\rho_X,\widehat\rho^{\rm good})
\stackrel{\rm law}{=}
\Fr(\rho_X,\widehat\rho^{\rm col});
\label{eq:fidelity-law-separation-main}
\end{equation}

\item their expected local log-fidelity losses are also identical and both satisfy the exact information-calibrated recovery budget,
\begin{equation}
-2\mathbb E\log\Fr(\rho_X,\widehat\rho^{\rm good})
=
-2\mathbb E\log\Fr(\rho_X,\widehat\rho^{\rm col})
\leq
\Delta,
\quad
\Delta=\frac12\log3;
\label{eq:twoq-budget-separation-main}
\end{equation}

\item their generated distributions are nevertheless macroscopically different:
\begin{equation}
\Wtr\left(\mu,\Law(\widehat\rho^{\rm good})\right)=0,
\quad
\Wtr\left(\mu,\Law(\widehat\rho^{\rm col})\right)
=
\frac{\sqrt5}{4}>\frac12,
\label{eq:wasserstein-separation-main}
\end{equation}
where
\begin{equation}
\mu=\frac1{16}\sum_{x=1}^{16}\delta_{\rho_x}
\end{equation}
is the target ensemble distribution.
\end{enumerate}
\end{theorem}

The theorem makes the distinction between the two objectives explicit inside the information-calibrated quantum diffusion setting. The $\Delta$-based criterion answers a local physical question: whether the reverse outputs are compatible with the recoverability scale induced by the forward information loss. It does not answer the generative question of whether the stochastic reverse process reproduces the correct probability law over quantum states. In the construction above, the latter can fail severely even though the complete local-fidelity statistics and the information-calibrated recovery budget are indistinguishable.

The construction is deliberately finite and genuinely quantum: it uses two qubits and pairwise noncommuting target states, so the separation is not an artifact of asymptotically large dimension or a purely classical commuting example. The explicit SIC-based construction and proof are given in Appendix~\ref{app:coverage-separation}.

The same separation can be made arbitrarily severe.

\begin{corollary}[Arbitrarily severe budget-feasible separation]
\label{cor:coverage-separation-amplified}
For every $c\in(0,1)$ there exists a finite-dimensional complete-depolarization problem and two stochastic CPTP reverse models with identical local-fidelity laws and identical expected local log-fidelity losses, both within the exact information-calibrated recovery budget, while one model has zero trace-Wasserstein error and the other has error larger than $c$. The dimension may be chosen as a power of two.
\end{corollary}

Thus the role of the learning objective is now separated into two parts. The information-calibrated constraint keeps each reverse step at a physically meaningful local recovery scale, while a distribution objective controls the probability law of the stochastic outputs. 

\section{From information-calibrated recovery to stochastic generation}
\label{sec:learning}
Section~\ref{sec:coverage-separation} shows that local recovery and generative coverage are non-substitutable. We therefore learn a stochastic reverse process with a distribution-matching objective, while using the information decrements $\Delta_t$ to constrain the quality of each local reverse step.

Let $K_{\theta,t}$ denote the learned stochastic CPTP reverse kernel at timestep $t$. For an input state $\rho$, $K_{\theta,t}(\rho,\cdot)$ is the output distribution obtained by sampling an input-independent latent variable and applying the corresponding physical branch. Given $X\sim p$, we write
\begin{equation}
\widehat\rho_{t-1}\sim K_{\theta,t}(\rho_{X,t},\cdot).
\label{eq:kernel-law-main}
\end{equation}
For a clipping level $\gamma\in(0,1)$, define the population local recovery loss
\begin{equation}
L_{\gamma,t}(\theta)=
\mathbb E\!\left[
-2\log\max\!\left\{\gamma,
\Fr(\rho_{X,t-1},\widehat\rho_{t-1})\right\}
\right].
\label{eq:population-clipped-loss-main}
\end{equation}
The clipping only keeps the log-fidelity loss bounded; it does not change the role of $\Delta_t$ as the physical recovery reference. By the common-channel budget theorem, the level $L_{\gamma,t}\leq\Delta_t$ is attainable in the unrestricted physical class.

For distribution matching, let
\begin{equation}
\widehat{\mathcal J}_{\rm dist}(\theta)
:=\frac1T\sum_{t=1}^T
\widehat{\mathrm{MMD}}_{\rm Gmix,t}^{2}(\theta),
\label{eq:distribution-objective-main}
\end{equation}
where $\widehat{\mathrm{MMD}}_{\rm Gmix,t}^{2}$ is the empirical Gaussian-mixture MMD$^2$ between the target state batch at timestep $t-1$ and the batch generated by $K_{\theta,t}$ from timestep-$t$ inputs; its kernel and estimator are given in Appendix~\ref{app:experimental-protocol}. Hats denote empirical quantities computed from the training data, including the estimated clock budget $\widehat\Delta_t$. The learned model is then obtained from
\begin{equation}
\min_\theta\widehat{\mathcal J}_{\rm dist}(\theta)
\quad\text{s.t.}\quad
\widehat L_{\gamma,t}(\theta)\leq\widehat\Delta_t,
\qquad t=1,\ldots,T.
\label{eq:constrained-training-main}
\end{equation}
Thus the objective controls stochastic coverage, while the constraints keep every local reverse step on the information-calibrated recovery scale.

\begin{theorem}[Uniform finite-sample calibration]
\label{thm:generalization}
Consider the bounded-generator Stinespring class in Appendix~\ref{app:generalization}, with $P$ trainable rotations and $n$ iid training examples. Suppose the trained parameter $\widehat\theta$ satisfies
\begin{equation}
\max_t\pos{\widehat L_{\gamma,t}(\widehat\theta)-\widehat\Delta_t}
\leq\eta_{\rm train}.
\end{equation}
Define the clock-estimation error
\begin{equation}
\varepsilon_{\rm clk}
:=\max_t\left|\widehat\Delta_t-\Delta_t\right|
\end{equation}
and the uniform generalization term
\begin{equation}
\varepsilon_{\rm gen}
:=\frac{4\log(1/\gamma)}{\sqrt n}
+2\log(1/\gamma)
\sqrt{
\frac{
P\log\!\left(1+\frac{2\pi Pn}{\gamma^2\log^2(1/\gamma)}\right)
+\log(T/\delta)
}{2n}}.
\label{eq:epsilon-gen-main}
\end{equation}
Then, with probability at least $1-\delta$,
\begin{equation}
L_{\gamma,t}(\widehat\theta)
\leq
\Delta_t+\varepsilon_{\rm cal}
\qquad\text{for every }t,
\label{eq:population-calibration-main}
\end{equation}
where
\begin{equation}
\varepsilon_{\rm cal}
:=\eta_{\rm train}+\varepsilon_{\rm clk}+\varepsilon_{\rm gen}.
\label{eq:epsilon-cal-main}
\end{equation}
\end{theorem}
The three terms in $\varepsilon_{\rm cal}$ respectively account for incomplete constraint satisfaction, error in the estimated information clock, and finite-sample generalization. The theorem is a validity guarantee; in the finite-support experiments we also evaluate the corresponding population calibration quantities directly.

The preceding result converts the empirical constraints into population-level local guarantees. To connect those guarantees to the final generated distribution, let
\begin{equation}
\mu_t=\sum_xp_x\delta_{\rho_{x,t}}
\end{equation}
be the target distribution at timestep $t$, and let $\nu_t$ be the distribution produced by the learned reverse chain. A local bound $L_{\gamma,t}\leq a$ implies a one-step trace-Wasserstein error bounded by the transfer function
\begin{equation}
\psi_\gamma(a)=
\sqrt{1-\pos{e^{-a/2}-\frac{\gamma a}{2\log(1/\gamma)}}^{\,2}}.
\label{eq:psi-main}
\end{equation}
Appendix~\ref{app:endpoint} derives this conversion from clipped log-fidelity to trace distance. Composing the resulting one-step bounds gives the endpoint guarantee.

\begin{theorem}[Compositional endpoint bound]
\label{thm:endpoint-certificate}
For the input-independent latent-conditioned CPTP kernels above,
\begin{equation}
\Wtr(\mu_0,\nu_0)
\leq
\Wtr(\mu_T,\nu_T)
+\sum_{t=1}^T\psi_\gamma(\Delta_t+\varepsilon_{\rm cal}).
\label{eq:endpoint-simple-main}
\end{equation}
For a maximally mixed terminal prior,
\begin{equation}
\Wtr(\mu_T,\nu_T)\leq\lambda_T(1-1/d).
\label{eq:terminal-mismatch-main}
\end{equation}
\end{theorem}
The theorem propagates the calibrated local errors through the stochastic reverse chain in the same trace-Wasserstein geometry used for endpoint evaluation. It is a one-sided certificate, not a substitute for distribution matching: coverage remains an independent requirement by \cref{thm:coverage-separation-main}. Stronger contraction-weighted forms and finite-resource extensions are given in Appendices~\ref{app:endpoint}, \ref{app:finite-shot}, and~\ref{app:architecture}.

\section{Experiments}
\label{sec:experiments}

The experiments test the theory in the same order in which it was
developed. We first ask whether the information decrement $\Delta_t$
tracks the realized difficulty of individual reverse steps and whether
enforcing the corresponding budgets improves local recovery. We then
test the recoverability--coverage separation, evaluate the complete
generator at fixed model capacity, study the effect of increasing
reverse-model expressivity, and finally isolate the contribution of the
equal-information schedule on an independently seeded cohort.

To make the comparisons explicit, \cref{tab:experiment-variants-main}
summarizes the main variants and the role of each comparison. A--E use
the same compact depth-$8$ reverse architecture in the primary
controlled comparisons; R is used only for the
recoverability--coverage ablation, and F only for the independently
seeded constrained schedule comparison.
\begin{table}[htbp]
\centering
\caption{\textbf{Experimental variants.} ``Dist.'' denotes the
distribution-matching objective and ``constraints'' denotes the
timestep-wise cq-information recovery budgets.}
\label{tab:experiment-variants-main}
\begin{tabular}{llll}
\toprule
Variant & Forward grid & Reverse training rule & Main role\\
\midrule
A & linear retention & Dist. only & baseline for C\\
B & equal cq-MI & Dist. only & baseline for D\\
C & linear retention & Dist. + constraints & recovery-budget intervention\\
D & equal cq-MI & Dist. + constraints & complete base method\\
E & cosine retention & Dist. only & heterogeneous-schedule diagnostic\\
F & cosine retention & Dist. + constraints & constrained schedule comparison\\
R & equal cq-MI & local criterion only & coverage ablation\\
\bottomrule
\end{tabular}
\end{table}

The controlled two-qubit regime uses three entangled mixed-state
difficulty levels and ten matched seeds ($150$ runs). The primary
four-qubit TFIM study uses $100$ training and $100$ held-out states for
each of ten seeds; held-out states are used only for generation metrics.
A frozen broader-field study stress-tests the local calibration
mechanism. Full protocols are in
Appendix~\ref{app:experimental-protocol}.

\subsection{Calibration mechanism}
\label{sec:exp-twoq}

This subsection tests the two claims underlying the local calibration
mechanism. First, before imposing any recovery constraint, does
$\Delta_t$ order the reverse-step loss realized along the forward path?
Second, at a fixed forward grid, does explicitly enforcing
$L_{\gamma,t}\leq\Delta_t$ improve an independent local-recovery
diagnostic?

For the first question, we evaluate the rank association between
$\Delta_t$ and realized reverse-step loss across heterogeneous
unconstrained schedules. For the second, C--A and D--B are matched
interventions: each pair uses the same forward grid and distributional
objective, while only the latter member adds the cq-information recovery
constraints. The cross-scale numerical results are summarized in
\cref{tab:calibration-summary}, while
\cref{fig:cross-scale-mechanism} visualizes the corresponding mechanism:
panel (a) shows information-loss/reverse-loss alignment, panel (b) shows
the reduction in budget excess under constraint enforcement, and panel
(c) relates remaining excess to local and endpoint errors.
\begin{table}[htbp]
\centering
\caption{\textbf{Calibration mechanism across scales.} Controlled 2q reports matched improvement frequency; primary 4q reports reduction relative to the matched unconstrained mean.}
\label{tab:calibration-summary}
\begin{tabular}{lcc}
\toprule
Check & Controlled 2q & Primary 4q\\
\midrule
$\Delta_t$ vs. reverse loss & $\rho=.800$; 100\% positive & $\rho=.944$; 100\% positive\\
C vs. A: max local error & 83.3\% improve & 41.0\% reduction; 100\% improve\\
D vs. B: max local error & 100\% improve & 26.6\% reduction; 100\% improve\\
\bottomrule
\end{tabular}
\end{table}

\begin{figure}[htbp]
  \centering
  \includegraphics[width=.78\linewidth]{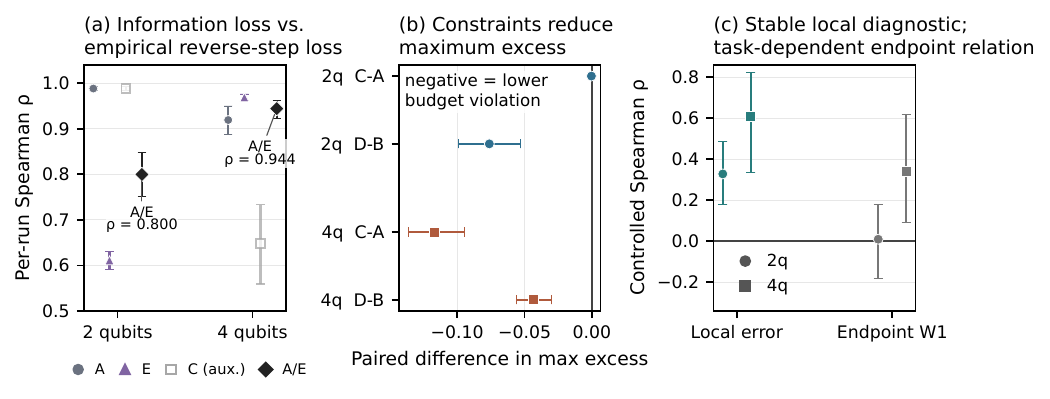}
  \caption{\textbf{Cross-scale calibration mechanism.} (a) Information-loss alignment on unconstrained A/E runs. (b) Budget enforcement reduces maximum excess. (c) Excess remains associated with local error, whereas its endpoint association is task-dependent.}
  \label{fig:cross-scale-mechanism}
\end{figure}

As summarized in \cref{tab:calibration-summary} and illustrated in
\cref{fig:cross-scale-mechanism}, the results consistently validate both
parts of the local calibration mechanism across scales: larger
information decrements are associated with larger realized reverse
losses, while enforcing the information-calibrated budgets
substantially reduces maximum local trace error on matched forward
grids. The controlled interventions are particularly strong on the
primary four-qubit task, where C and D reduce maximum local error by
$41.0\%$ and $26.6\%$, respectively, with improvement on all matched
seeds.

The correlation result is interpreted specifically as evidence for the
predicted ordering along the realized forward paths. Because the
eight-step schedules are largely rank-collinear with simple monotone
noise coordinates, we do not use it as an independent-identifiability
claim beyond timestep or raw noise strength; Appendix~\ref{app:proxy-audit}
gives the corresponding audit.

Independent numerical checks further reinforce the calibration
mechanism. The exact recovery SDP is solved on all $72$ controlled
two-qubit timestep instances, and every numerical optimum lies inside
the continuity-based component of the information-theoretic recovery
bracket within tolerance $7\times10^{-5}$. A frozen broader-field shift
preserves the local calibration mechanism, while the finite-ensemble
endpoint-certificate audit reports no numerical violations across the
retained two-/four-qubit runs. These checks support the robustness of
the local recovery calibration across the tested regimes; details are
in Appendices~\ref{app:validation-diagnostics},
\ref{app:broader-field-results}, and~\ref{app:certificate-audit}.

\subsection{Recoverability--coverage ablation}
\label{sec:exp-coverage-ablation}

The previous subsection tests local calibration. We next test the
distinct claim from \cref{thm:coverage-separation-main}: matching local
recovery need not determine endpoint coverage. R and D use the same
equal-cq-MI grid, data, initialization, stochastic architecture,
optimization length, and endpoint evaluation. R optimizes only the
local criterion, whereas D combines distribution matching with the
cq-information constraints. B provides the complementary
distribution-only control.
\begin{table}[htbp]
\centering
\caption{\textbf{Recoverability--coverage ablation.} Mean $\pm$ s.e. over ten matched seeds; lower is better.}
\label{tab:coverage-ablation-main}
\begin{tabular}{lcc}
\toprule
Variant & Max local trace error & Endpoint $\Wtr$\\
\midrule
B: distribution-only & $.448\!\pm\!.019$ & $.702\!\pm\!.016$\\
R: local-only & $.329\!\pm\!.003$ & $.691\!\pm\!.023$\\
D: combined & $\mathbf{.329\!\pm\!.002}$ & $\mathbf{.622\!\pm\!.036}$\\
\bottomrule
\end{tabular}
\end{table}

As shown in \cref{tab:coverage-ablation-main}, R and D have essentially
identical maximum local trace error ($.328685$ versus $.328600$), yet D
lowers mean endpoint $\Wtr$ by 10.0\% and wins $8/10$ matched seeds.
Thus comparable local recovery does not determine the generated state
distribution. The same table shows the complementary role of the local
constraint: B, which uses distribution matching alone, has substantially
larger local error. The ablation therefore supports the two-part
learning objective of \cref{sec:learning}: recovery constraints and
distribution matching control different aspects of the reverse model.

\subsection{Base-capacity generation}
\label{sec:exp-fourq}

Having tested the two mechanisms separately, we next evaluate the
complete designs under a common model-capacity budget. All A--E
variants use the same compact depth-$8$ architecture, so differences in
Table~\ref{tab:fourq-summary} reflect schedule and training-rule choices
rather than reverse-model size.
\begin{table}[htbp]
\centering
\caption{\textbf{Primary four-qubit TFIM.} Mean $\pm$ s.e. over ten seeds. Diversity ratio is generated/target pairwise trace-distance diversity.}
\label{tab:fourq-summary}
\resizebox{\linewidth}{!}{\begin{tabular}{lrrrrrr}
\toprule
& Excess & Max local & Endpoint $\Wtr$ & HS-MMD$^2$ & Obs. err. & Div. ratio\\
\midrule
A & $.286\!\pm\!.011$ & $.536\!\pm\!.019$ & $.640\!\pm\!.026$ & $.435\!\pm\!.035$ & $.397\!\pm\!.016$ & $3.13\!\pm\!.15$\\
B & $.187\!\pm\!.005$ & $.448\!\pm\!.019$ & $.702\!\pm\!.016$ & $.528\!\pm\!.024$ & $.417\!\pm\!.010$ & $3.05\!\pm\!.11$\\
C & $.169\!\pm\!.002$ & $\mathbf{.316\!\pm\!.005}$ & $.671\!\pm\!.019$ & $.479\!\pm\!.027$ & $.423\!\pm\!.015$ & $2.30\!\pm\!.07$\\
D & $\mathbf{.143\!\pm\!.005}$ & $.329\!\pm\!.002$ & $\mathbf{.622\!\pm\!.036}$ & $\mathbf{.420\!\pm\!.046}$ & $\mathbf{.385\!\pm\!.022}$ & $2.39\!\pm\!.07$\\
E & $.202\!\pm\!.007$ & $.505\!\pm\!.009$ & $.721\!\pm\!.015$ & $.550\!\pm\!.023$ & $.443\!\pm\!.009$ & $2.83\!\pm\!.08$\\
\bottomrule
\end{tabular}}
\end{table}

The full comparison in \cref{tab:fourq-summary} makes the
local-versus-global distinction explicit. C achieves the smallest
maximum local trace error, whereas D attains the lowest mean endpoint
$\Wtr$ and the best mean HS-MMD$^2$ and observable error among A--E.
Thus the design with the strongest local inversion need not be the one
with the best endpoint coverage. Under the controlled depth-$8$
capacity, the combined equal-cq-MI design D provides the strongest
overall generator, while C remains the best purely local-recovery
result.

\subsection{Capacity extension and external QuDDPM benchmark}
\label{sec:exp-capacity-external}

The controlled depth-$8$ comparisons establish the mechanism under a
shared compact architecture. We next test how far the same
information-calibrated design improves when reverse-model expressivity
is increased while all other scientific choices are held fixed.
Keeping $T=8$, the equal-cq-MI grid and budgets, latent support,
objective, optimizer, data, and evaluation unchanged, we increase only
the shared Stinespring depth of D from $8$ to $128$. This controlled
capacity extension raises the model from $144$ to $2304$ trainable
rotations and reduces endpoint $\Wtr$ from $.622$ to $.424$ (31.8\%),
with improvement on all $10/10$ matched seeds and paired bootstrap
interval $[-.262,-.128]$. Local recovery improves in parallel:
maximum local trace error falls from $.329$ to $.287$, and maximum
population excess from $.143$ to $.092$. These gains identify
reverse-model expressivity as a major performance axis within the
fixed information-calibrated framework. We use the term
\emph{capacity extension} rather than infer a general scaling law from
two model sizes.

We then compare the capacity-extended model with the official QuDDPM
implementation on the identical ten TFIM datasets and held-out
evaluator; the headline endpoint comparison is summarized in
\cref{tab:capacity-external-main}. On endpoint trace-Wasserstein error,
the geometry directly matched to our local-to-global theory, our
depth-$128$ learner improves over official QuDDPM from $.498$ to $.424$
on all $10/10$ matched datasets, with paired interval
$[-.083,-.064]$, while using $2304$ rather than $4320$ trainable
parameters. This provides a strong external endpoint benchmark for the
complete information-calibrated learner. The comparison is
metric-specific: complementary HS-MMD$^2$, TFIM observable error, and
runtime favor QuDDPM and are reported with the full comparison profile
in Appendix~\ref{app:capacity-external}.
\begin{table}[htbp]
\centering
\caption{\textbf{Capacity extension and external endpoint benchmark.} Mean endpoint $\Wtr$ over the ten matched primary TFIM seeds; lower is better.}
\label{tab:capacity-external-main}
\begin{tabular}{lrrr}
\toprule
Method & Shared depth & Trainable params. & Endpoint $\Wtr$\\
\midrule
D (base) & $8$ & $144$ & $.622$\\
D (capacity-scaled) & $128$ & $2304$ & $\mathbf{.424}$\\
QuDDPM & -- & $4320$ & $.498$\\
\bottomrule
\end{tabular}
\end{table}

Taken together, the capacity comparison in
\cref{tab:capacity-external-main} and a matched $T=30$ control with the
compact shared model isolate reverse-model expressivity as the effective
axis in this experiment: increasing depth produces the large endpoint
gain, whereas finer discretization leaves endpoint error essentially
unchanged.

\subsection{Independent schedule comparison}
\label{sec:exp-schedule-confirm}

The preceding experiments establish the effects of recovery-budget
enforcement and reverse-model capacity. We finally isolate the
contribution of the forward schedule itself. On a fresh cohort of
dataset seeds $100,\ldots,109$, we hold the architecture, constrained
training rule, optimizer, and endpoint evaluation fixed and change only
the forward grid. The comparison uses C (linear retention), D (equal
cq-MI), and F (cosine retention), all with the same recovery constraints.

\begin{table}[htbp]
\centering
\caption{\textbf{Independently seeded matched constrained schedule comparison at $T=8$.} Endpoint $\Wtr$ is mean $\pm$ s.e. over dataset seeds $100,\ldots,109$. ``Wins'' counts matched seeds on which D has lower endpoint error; intervals report paired D-minus-baseline bootstrap intervals.}
\label{tab:schedule-confirm-main}
\begin{tabular}{lcccc}
\toprule
Schedule & Endpoint $\Wtr$ & D improvement & D wins & Paired interval \\
\midrule
C: linear & $.683\!\pm\!.024$ & $5.1\%$ & $7/10$ & $[-.075,.008]$\\
D: equal-cq-MI & $\mathbf{.648\!\pm\!.022}$ & -- & -- & --\\
F: cosine & $.739\!\pm\!.016$ & $\mathbf{12.3\%}$ & $9/10$ & $[-.129,-.043]$\\
\bottomrule
\end{tabular}
\end{table}

As shown in \cref{tab:schedule-confirm-main}, the independently seeded
comparison directly supports the scheduling principle itself rather
than the effect of model capacity or constraint enforcement. Equal-cq-MI
allocation gives a clear advantage over cosine scheduling under the
matched constrained protocol. Relative to linear retention scheduling,
the result is directionally consistent in favor of equal-cq-MI, although
the paired interval still crosses zero. The schedule experiment therefore
provides independent empirical support for allocating the forward path in
ensemble-information coordinates rather than by a standard heuristic
noise schedule.

\section{Discussion and conclusion}

We introduced an information-calibrated view of quantum diffusion in
which the forward process is organized by ensemble information loss
rather than raw channel strength. Along depolarization, the cq-information
decrement $\Delta_t$ provides a single operational coordinate with two
roles: equalizing it gives the unique minimax discretization of the
forward information loss, while universal recoverability turns the same
quantity into a physically attainable scale for local reverse recovery.
The recoverability--coverage separation then identifies the additional
requirement introduced by stochastic generation: matching the local
physical inverse does not by itself determine how probability mass is
allocated across the generated ensemble. This leads to a complete learner
in which information-calibrated recovery constraints and distribution
matching play complementary, non-substitutable roles.

The experiments validate this decomposition at the mechanism and
generation levels. Across controlled two- and four-qubit settings,
$\Delta_t$ consistently orders realized reverse-step difficulty and
budget enforcement improves independent local-recovery diagnostics. The
recoverability--coverage ablation shows that comparable local recovery
can coexist with materially different endpoint generation quality, while
the combined learner gives the strongest controlled generator at fixed
base capacity. Increasing reverse-model expressivity produces a further
substantial endpoint improvement, and the capacity-extended model attains
lower endpoint trace-Wasserstein error than the official QuDDPM
implementation while using fewer trainable parameters. Independently
seeded experiments further confirm the benefit of equal-information
allocation over cosine scheduling and show a favorable trend over linear
retention scheduling.

More broadly, our results suggest a different way to design future
quantum diffusion models. Forward corruption need not be treated only as
a convenient path from data to a simple terminal state; it can be viewed
as a sequence of inverse problems whose difficulty should be calibrated
by an operational information quantity. This perspective separates three
design decisions that are often entangled: how information is removed in
the forward process, what level of local recovery is physically
reasonable at each reverse step, and how stochastic outputs should be
matched at the distribution level. Treating these roles separately
provides a principled alternative to choosing a noise schedule and a
training objective independently. In particular, information-calibrated
coordinates can serve as a common interface between forward-process
design, reverse-model training, and mechanism-level evaluation.

This viewpoint opens several natural directions for quantum generative
modeling. The cq-information clock can be extended beyond depolarization
to other physically relevant quantum channels whenever the retained
ensemble information can be estimated and related to recoverability,
suggesting adaptive schedules tailored to dephasing, amplitude damping,
hardware noise, or learned forward processes. The same principle also
invites resource-aware implementations in which information budgets are
estimated from finite-copy measurements and used to allocate reverse
model capacity or experimental effort preferentially to harder steps.
At larger system sizes, scalable estimators such as structured
tomography or shadow-based methods could replace full state
reconstruction while preserving the information-calibration principle.
Finally, the separation between local recoverability and global coverage
suggests a broader modular design pattern for quantum generative models:
use operational quantum-information quantities to calibrate what each
local transformation should achieve, and use distributional objectives
only for the stochastic degrees of freedom that local recovery cannot
identify.

Taken together, the central message is that the forward process, the
reverse physical task, and the generative distribution need not be
designed as separate components. Classical--quantum information provides
a bridge between them: it turns forward corruption into an interpretable
diffusion clock, gives reverse recovery a theorem-determined scale, and
clarifies exactly where distribution-level learning must enter. We expect
this information-calibrated perspective to provide a useful foundation
for designing and analyzing future quantum diffusion and related quantum
generative models.

\clearpage
\paragraph{Reproducibility statement.}
Complete proofs and experimental protocols are in the appendix. The anonymized code package contains data generators, schedule construction, Stinespring models, primal--dual training, fixed configurations/seeds, saved run-level results, capacity-extension runs, the audited external QuDDPM evaluation, and scripts that regenerate the reported tables and figures. No held-out endpoint accuracy metric such as $\Wtr$, HS-MMD$^2$, or the TFIM observable is used for architecture or checkpoint selection; the disjoint architecture pilot uses a coarse diversity/non-collapse gate only to exclude degenerate candidates.

\section*{AI Use Statement}
Large language models were used as general-purpose research assistance for language polishing, literature discovery, exploring mathematical formulations and proof strategies, experimental design, and generating portions of the experimental code. All mathematical statements and proofs included in the paper were checked and accepted by the authors; references surfaced with AI assistance were verified against the original sources before citation. AI-assisted code was manually reviewed, tested, and checked against the declared experimental protocol before execution, and all reported numerical values were obtained from executed experiments rather than generated by a language model. The authors take full responsibility for the paper's claims, proofs, citations, code, and experimental results.

\clearpage
\bibliographystyle{plainnat}
\bibliography{refs_iclr}

\clearpage
\appendix

\section{CQ identities and the forward information clock}
\label{app:forward}
The main text uses two properties of the forward coordinate: the cq identity and invertibility of the depolarizing curve. This section proves both before any recovery argument is invoked.

\subsection{The cq mutual information equals the Holevo quantity}
For $\omega^{XQ}=\sum_xp_x\ketbra{x}\otimes\rho_x$, block diagonalization gives
\begin{equation}
  S(\omega^{XQ})=H(p)+\sum_xp_xS(\rho_x),
\end{equation}
where $H(p)=-\sum_xp_x\log p_x$ is the Shannon entropy. Moreover, $S(\omega^X)=H(p)$ and $S(\omega^Q)=S(\bar\rho)$. Substitution into the definition of quantum mutual information in the introduction proves \cref{eq:holevo-identity}.

\subsection{Proof of the forward-clock theorem}
Write $\rho_x(\lambda):=\cD_\lambda(\rho_x)$ and $\bar\rho(\lambda):=\sum_xp_x\rho_x(\lambda)$. Continuity follows from finite-dimensional entropy continuity. For $0<\lambda<1$, all $\rho_x(\lambda)$ and $\bar\rho(\lambda)$ are positive definite, so analyticity follows from analyticity of the matrix logarithm on the positive cone.

For $\lambda=\alpha\lambda_1+(1-\alpha)\lambda_2$, both arguments of each relative entropy in \cref{eq:chi-curve-main} are the corresponding convex combinations. Joint convexity of quantum relative entropy therefore yields
\begin{equation}
  \chi(\lambda)
  \leq\alpha\chi(\lambda_1)+(1-\alpha)\chi(\lambda_2).
\end{equation}
For $a\in[0,1]$,
\begin{equation}
  (\rho_x(a\lambda),\bar\rho(a\lambda))
  =a(\rho_x(\lambda),\bar\rho(\lambda))
  +(1-a)(\id/d,\id/d).
\end{equation}
Joint convexity therefore gives the contraction relation
\begin{equation}
  \chi(a\lambda)\leq a\chi(\lambda).
  \label{eq:linear-contraction-app}
\end{equation}

For $\lambda>0$, the channel $\cD_\lambda$ is injective on trace-one operators. A nontrivial ensemble therefore remains nontrivial, so $\chi(\lambda)>0$. If $0<\mu<\lambda$, apply \cref{eq:linear-contraction-app} with $a=\mu/\lambda$:
\begin{equation}
  \chi(\mu)\leq\frac\mu\lambda\chi(\lambda)<\chi(\lambda).
\end{equation}
Together with $\chi(0)=0$, this proves strict increase. The inverse grid consequently exists and is unique.

For any monotone grid with fixed endpoints,
\begin{equation}
  \sum_{t=1}^T\Delta_t=\Delta_{\rm tot},
\end{equation}
so $\max_t\Delta_t\geq \Delta_{\rm tot}/T$, with equality if and only if all decrements are equal. Strict monotonicity of $\chi$ then gives uniqueness of the grid. Jensen's inequality gives
\begin{equation}
  \sum_{t=1}^T\varphi(\Delta_t)
  \geq T\varphi(\Delta_{\rm tot}/T),
\end{equation}
with equality only at the equal vector when $\varphi$ is strictly convex.
\hfill$\square$

\section{Learner-facing and two-sided common-channel recovery}
\label{app:recovery}
We use the following notation for the recovery proofs. Write $\chi_t=\chi(\lambda_t)$, $\bar\rho_t=\sum_xp_x\rho_{x,t}$, and, whenever $\lambda_{t-1}>0$, $\cN_t=\cD_{\lambda_t/\lambda_{t-1}}$. Also write
\begin{align}
\cL_t^{\log}(\cR)&:=-2\sum_xp_x\log\Fr(\rho_{x,t-1},\cR(\rho_{x,t})),\\
\cL_t^{\rm cq}(\cR)&:=-2\log\Fr\!\left(\omega_{t-1}^{XQ},(\idmap_X\otimes\cR)(\omega_t^{XQ})\right),
\end{align}
where $\omega_t^{XQ}=\sum_xp_x\ketbra{x}\otimes\rho_{x,t}$ and
\begin{equation}
\Fr\!\left(\omega_{t-1}^{XQ},(\idmap_X\otimes\cR)(\omega_t^{XQ})\right)
=\sum_xp_x\Fr(\rho_{x,t-1},\cR(\rho_{x,t})).
\label{eq:cq-fidelity-block-main}
\end{equation}
Let $m=|\mathcal X|$. With $h_2(r)=-r\log r-(1-r)\log(1-r)$ and $0\log0:=0$, define
\begin{align}
g(r)&=(1+r)\log(1+r)-r\log r,\\
\mathfrak f_d(r)&=h_2(\min\{r,1-1/d\})+\min\{r,1-1/d\}\log(d-1),
\label{eq:continuity-functions-main}
\end{align}
and the same-prior Holevo continuity modulus
\begin{equation}
\Omega_{d,m}(r)=\min\!\left\{2\mathfrak f_d(r),\,r\log\min\{d,m\}+g(r),\,\log\min\{d,m\}\right\},
\label{eq:holevo-continuity-modulus-main}
\end{equation}
with generalized inverse $\Omega_{d,m}^{\leftarrow}(a)=\inf\{r\in[0,1]:\Omega_{d,m}(r)\ge a\}$. For two ensembles on the same index set, Shirokov's distance is
\begin{equation}
D_{\rm ens}(\{p_x,\rho_x\},\{q_x,\sigma_x\})
=\frac12\sum_x\|p_x\rho_x-q_x\sigma_x\|_1.
\label{eq:ensemble-distance-main}
\end{equation}

This appendix separates the ingredients of the local recovery theory. Ensemble-level Holevo-loss recoverability is established in prior work \citep{buscemi2016approximate}, and the universal recovery theorem supplies a map depending only on the reference state and forward channel \citep{junge2018universal}. We first show that this same map controls the stricter average-log risk used by our learner. We then combine data processing with established Holevo-continuity bounds and, separately, derive the pairwise-geometric converse used in the strengthened lower bracket. A decision-value witness from the statistical-comparison viewpoint also yields an exact symmetric orthogonal benchmark.

\subsection{Proof of the learner-facing common-channel budget}
The universal recovery theorem states that, for states $\rho,\sigma$ and a quantum channel $\cN$, there is a recovery channel $\cR_{\sigma,\cN}^{\mathrm U}$ depending only on $(\sigma,\cN)$ such that
\begin{equation}
  D(\rho\Vert\sigma)-D(\cN(\rho)\Vert\cN(\sigma))
  \geq-2\log\Fr\bigl(\rho,(\cR_{\sigma,\cN}^{\mathrm U}\circ\cN)(\rho)\bigr),
  \label{eq:universal-recovery}
\end{equation}
and $(\cR_{\sigma,\cN}^{\mathrm U}\circ\cN)(\sigma)=\sigma$ \citep{junge2018universal}.

Because $p_x>0$ and $\bar\rho_{t-1}=\sum_xp_x\rho_{x,t-1}$, one has $\supp\rho_{x,t-1}\subseteq\supp\bar\rho_{t-1}$ for every $x$. Apply \cref{eq:universal-recovery} to every pair $(\rho_{x,t-1},\bar\rho_{t-1})$ with the same channel $\cN_t$. Universality means that the same recovery channel, depending only on $(\bar\rho_{t-1},\cN_t)$, works for all $x$. Moreover,
\begin{align}
  \Delta_t
  =\sum_xp_x\Bigl[
    D(\rho_{x,t-1}\Vert\bar\rho_{t-1})
    -D(\rho_{x,t}\Vert\bar\rho_t)
  \Bigr].
  \label{eq:delta-relative-entropy-expansion}
\end{align}
Let
\begin{equation}
  F_x:=\Fr\bigl(\rho_{x,t-1},\cR_t^{\mathrm U}(\rho_{x,t})\bigr).
\end{equation}
Summing with weights $p_x$ gives
\begin{equation}
  \Delta_t\geq-2\sum_xp_x\log F_x=\cL_t^{\log}(\cR_t^{\mathrm U}).
\end{equation}
Since $-\log$ is convex,
\begin{equation}
  -2\log\sum_xp_xF_x
  \leq-2\sum_xp_x\log F_x.
\end{equation}
Combining this inequality with \cref{eq:log-budget-main} and exponentiating gives the appendix-local consequence
\begin{equation}
  \sum_xp_xF_x\geq e^{-\Delta_t/2}.
  \label{eq:average-fidelity-app}
\end{equation}
By \cref{eq:cq-fidelity-block-main}, the same weighted sum is the root fidelity of the corresponding cq block states.
\hfill$\square$

\subsection{Proof of the two-sided intrinsic recoverability theorem}
Fix a timestep and an arbitrary label-independent CPTP recovery channel $\cR$. Write
\begin{equation}
 \sigma_x:=\cR(\rho_{x,t}),\qquad
 \cE_{\cR}:=\{p_x,\sigma_x\},\qquad
 \varepsilon(\cR):=\sum_xp_x\dtr(\rho_{x,t-1},\sigma_x).
 \label{eq:recovered-ensemble-app}
\end{equation}
Because $\cE_t$ is mapped to $\cE_{\cR}$ by $\cR$, data processing for the Holevo quantity gives $\chi(\cE_{\cR})\leq\chi_t$. Hence
\begin{equation}
 \Delta_t=\chi_{t-1}-\chi_t
 \leq \chi_{t-1}-\chi(\cE_{\cR}).
 \label{eq:converse-dpi-app}
\end{equation}
The two ensembles on the right have the same probabilities. Their Shirokov ensemble distance $D_{\rm ens}$ from \cref{eq:ensemble-distance-main} is exactly $\varepsilon(\cR)$. Proposition~5, Eqs.~(60)--(61), of \citet{shirokov2017continuity} gives separate bounds with coefficients $\log d$ and $\log m$; the same proposition states that when the priors coincide, the generic $2g(D_{\rm ens})$ remainder can be replaced by $g(D_{\rm ens})$. Hence
\begin{equation}
 \left|\chi_{t-1}-\chi(\cE_{\cR})\right|
 \leq \varepsilon(\cR)\log\min\{d,m\}+g(\varepsilon(\cR)).
 \label{eq:shirokov-holevo-app}
\end{equation}
An independent entropy-continuity route yields the other branch in \cref{eq:holevo-continuity-modulus-main}. Let $\bar\sigma=\sum_xp_x\sigma_x$ and $\varepsilon_x=\dtr(\rho_{x,t-1},\sigma_x)$. Convexity of trace distance gives $\dtr(\bar\rho_{t-1},\bar\sigma)\leq\varepsilon(\cR)$. Audenaert continuity and concavity of the monotone modulus $\mathfrak f_d$ then give
\begin{align}
 \left|\chi_{t-1}-\chi(\cE_{\cR})\right|
 &\leq \mathfrak f_d(\varepsilon(\cR))
 +\sum_xp_x\mathfrak f_d(\varepsilon_x)\\
 &\leq 2\mathfrak f_d(\varepsilon(\cR)).
 \label{eq:audenaert-holevo-app}
\end{align}
Finally, every Holevo quantity lies in $[0,\log\min\{d,m\}]$. Combining these three valid bounds with \cref{eq:converse-dpi-app} proves, for every $\cR$,
\begin{equation}
 \Delta_t\leq\Omega_{d,m}(\varepsilon(\cR)).
 \label{eq:converse-modulus-app}
\end{equation}
In finite dimension the set of CPTP maps is compact in the Choi representation, and the objective in \cref{eq:intrinsic-recovery-error-main} is continuous, so the infimum is attained by some $\cR_t^\star$. Applying \cref{eq:converse-modulus-app} to $\cR_t^\star$ gives the main-text converse \cref{eq:information-recovery-converse-main}.

For the upper bound, use the universal channel $\cR_t^{\mathrm U}$ from \cref{thm:recovery-budget}. With $F_x:=\Fr(\rho_{x,t-1},\cR_t^{\mathrm U}(\rho_{x,t}))$, \cref{eq:average-fidelity-app}, the upper Fuchs--van de Graaf inequality in \cref{eq:fvdg-main}, and concavity of $f\mapsto\sqrt{1-f^2}$ give
\begin{align}
 \sum_xp_x\dtr(\rho_{x,t-1},\cR_t^{\mathrm U}(\rho_{x,t}))
 &\leq \sum_xp_x\sqrt{1-F_x^2}\\
 &\leq \sqrt{1-\left(\sum_xp_xF_x\right)^2}\\
 &\leq \sqrt{1-e^{-\Delta_t}}.
\end{align}
Taking the infimum over common channels proves \cref{eq:information-recovery-upper-main}.

The equivalence in \cref{eq:information-recovery-equivalence-main} now follows directly. If $\Delta_t\to0$, the upper bound gives $\varepsilon_t^\star\to0$. Conversely, \cref{eq:information-recovery-converse-main} and continuity of $\Omega_{d,m}$ at the origin, with $\Omega_{d,m}(0)=0$, imply that $\varepsilon_t^\star\to0$ forces $\Delta_t\to0$.
\hfill$\square$

\subsection{Proof of the pairwise-geometric converse}
\label{app:pairwise-proof}
For each pair $x<y$, define the distinguishability erased by the $t$th forward step,
\begin{equation}
\delta_{xy}^{(t)}
:=
\dtr(\rho_{x,t-1},\rho_{y,t-1})
-
\dtr(\rho_{x,t},\rho_{y,t})
\geq0,
\label{eq:pairwise-loss-app}
\end{equation}
where nonnegativity follows from contractivity of trace distance under the forward CPTP map. Define the geometry-only relaxation
\begin{equation}
\varepsilon_t^{\rm geo}
:=
\min_{r_x\geq0}
\left\{
\sum_xp_xr_x:
\ r_x+r_y\geq\delta_{xy}^{(t)}\ \forall x<y
\right\}.
\label{eq:pairwise-lp-app}
\end{equation}

Fix a label-independent CPTP recovery channel $\cR$ and set
\begin{equation}
e_x:=\dtr\bigl(\rho_{x,t-1},\cR(\rho_{x,t})\bigr).
\end{equation}
For any pair $x<y$, the triangle inequality gives
\begin{align}
\dtr(\rho_{x,t-1},\rho_{y,t-1})
&\leq e_x+
\dtr\bigl(\cR(\rho_{x,t}),\cR(\rho_{y,t})\bigr)+e_y\\
&\leq e_x+\dtr(\rho_{x,t},\rho_{y,t})+e_y,
\end{align}
where the second line again uses contractivity under CPTP maps. Hence
\begin{equation}
e_x+e_y\geq\delta_{xy}^{(t)}
\end{equation}
for every pair. The channel-induced vector $(e_x)_x$ is therefore feasible for \cref{eq:pairwise-lp-app}, so
\begin{equation}
\sum_xp_xe_x\geq\varepsilon_t^{\rm geo}.
\end{equation}
Taking the infimum over $\cR$ yields the channel-independent converse
\begin{equation}
\varepsilon_t^\star\geq\varepsilon_t^{\rm geo}.
\label{eq:pairwise-converse-app}
\end{equation}

For the depolarizing path,
\begin{equation}
\rho_{x,t}-\rho_{y,t}
=\lambda_t(\rho_x-\rho_y),
\end{equation}
so homogeneity of the trace norm yields
\begin{equation}
\dtr(\rho_{x,t},\rho_{y,t})
=\lambda_t\dtr(\rho_x,\rho_y).
\end{equation}
Therefore
\begin{equation}
\delta_{xy}^{(t)}
=(\lambda_{t-1}-\lambda_t)\dtr(\rho_x,\rho_y).
\end{equation}
Define the ensemble-geometric LP
\begin{equation}
\mathfrak G(\cE):=\min_{r_x\geq0}\left\{\sum_xp_xr_x:\ r_x+r_y\geq\dtr(\rho_x,\rho_y)\ \forall x<y\right\}.
\label{eq:ensemble-geometry-main}
\end{equation}
Positive homogeneity gives
\begin{equation}
\varepsilon_t^{\rm geo}=(\lambda_{t-1}-\lambda_t)\mathfrak G(\cE).
\label{eq:pairwise-factor-main}
\end{equation}
Finally, writing $d_{xy}=\dtr(\rho_x,\rho_y)$, standard LP duality gives
\begin{equation}
\mathfrak G(\cE)=
\max_{\alpha_{xy}\geq0}\left\{\sum_{x<y}\alpha_{xy}d_{xy}:\ \sum_{y\neq x}\alpha_{xy}\leq p_x\ \forall x\right\},
\label{eq:ensemble-geometry-dual-main}
\end{equation}
with $\alpha_{xy}=\alpha_{yx}$ understood on unordered pairs. Feasibility and boundedness are immediate, so strong duality applies.

For later use, monotonicity of $\Omega_{d,m}$ together with \cref{eq:converse-modulus-app} gives
$\varepsilon_t^\star\geq\Omega_{d,m}^{\leftarrow}(\Delta_t)$. Combining this with \cref{eq:pairwise-converse-app}, define
\begin{equation}
b_t:=\max\{\Omega_{d,m}^{\leftarrow}(\Delta_t),\varepsilon_t^{\rm geo}\},
\qquad
\varepsilon_t^\star\geq b_t.
\label{eq:recovery-lower-bracket-app}
\end{equation}
This appendix-only strengthening also yields a risk-level refinement. For an arbitrary $\cR$, let $d_x=\dtr(\rho_{x,t-1},\cR(\rho_{x,t}))$ and $F_x=\Fr(\rho_{x,t-1},\cR(\rho_{x,t}))$. The upper Fuchs--van de Graaf inequality and concavity of $r\mapsto\sqrt{1-r^2}$ give
\begin{equation}
\sum_xp_xF_x
\leq\sqrt{1-\left(\sum_xp_xd_x\right)^2}
\leq\sqrt{1-b_t^2}.
\end{equation}
Hence
\begin{equation}
-\log(1-b_t^2)
\leq\inf_{\cR}\cL_t^{\rm cq}(\cR)
\leq\inf_{\cR}\cL_t^{\log}(\cR)
\leq\Delta_t.
\label{eq:risk-bracket-app}
\end{equation}
\hfill$\square$

\subsection{Decision-value witness and the exact orthogonal benchmark}
\label{app:decision-witness}
\begin{corollary}[Exact orthogonal depolarizing benchmark]
\label{cor:orthogonal-exact-main}
For the uniform complete orthogonal ensemble
\[
\cE_\lambda^{(d)}=\left\{\frac1d,\rho_x^\lambda=\lambda\ketbra{x}+(1-\lambda)\frac{\id}{d}\right\}_{x=1}^d,
\]
the optimal common-channel average trace error for recovering retention $a$ from retention $b<a$ is
\begin{equation}
\varepsilon^\star(a\leftarrow b)=(a-b)\left(1-\frac1d\right).
\label{eq:orthogonal-exact-main}
\end{equation}
\end{corollary}
A simple operational witness follows from the randomization/deficiency viewpoint for quantum statistical experiments \citep{jencova2016comparison}. Let $\mathcal A$ be a finite action set and let $u(x,a)\in[0,1]$ be a bounded utility. Define the Bayes value
\begin{equation}
V_u(\cE):=\max_{\{M_a\}}
\sum_{x,a}p_xu(x,a)\Tr(M_a\rho_x),
\label{eq:bayes-value-app}
\end{equation}
where the maximum ranges over POVMs. Then every common recovery channel obeys
\begin{equation}
\sum_xp_x\dtr\bigl(\rho_{x,t-1},\cR(\rho_{x,t})\bigr)
\geq
\pos{V_u(\cE_{t-1})-V_u(\cE_t)}.
\label{eq:decision-witness-app}
\end{equation}
Indeed, let $\sigma_x=\cR(\rho_{x,t})$ and choose a POVM $\{M_a\}$ optimal for $\cE_{t-1}$. For each $x$, the operator $A_x:=\sum_a u(x,a)M_a$ satisfies $0\leq A_x\leq\id$, hence
\begin{equation}
\Tr[A_x(\rho_{x,t-1}-\sigma_x)]
\leq\dtr(\rho_{x,t-1},\sigma_x).
\end{equation}
Summing with $p_x$ gives
\begin{equation}
V_u(\cE_{t-1})-V_u(\{p_x,\sigma_x\})
\leq\sum_xp_x\dtr(\rho_{x,t-1},\sigma_x).
\end{equation}
Since $\{p_x,\sigma_x\}$ is a post-processing of $\cE_t$, its Bayes value is at most $V_u(\cE_t)$. This proves \cref{eq:decision-witness-app}. Choosing $u(x,a)=\mathbf 1\{a=x\}$ yields
\begin{equation}
\varepsilon_t^\star
\geq P_{\rm guess}(\cE_{t-1})-P_{\rm guess}(\cE_t).
\label{eq:guessing-witness-app}
\end{equation}

We now prove \cref{cor:orthogonal-exact-main}. For the uniform complete orthogonal ensemble $\cE_\lambda^{(d)}$, all states commute in the computational basis. Dephasing any POVM in this common eigenbasis leaves its outcome probabilities unchanged, reducing minimum-error discrimination to the corresponding classical problem; the maximum-likelihood basis decision is therefore optimal. Hence
\begin{equation}
P_{\rm guess}(\cE_\lambda^{(d)})
=\lambda+\frac{1-\lambda}{d}.
\end{equation}
Therefore \cref{eq:guessing-witness-app} gives
\begin{equation}
\varepsilon^\star(a\leftarrow b)
\geq(a-b)\left(1-\frac1d\right).
\end{equation}
For the matching upper bound, use the identity channel. Since
\begin{equation}
\rho_x^a-\rho_x^b
=(a-b)\left(\ketbra{x}-\frac{\id}{d}\right),
\end{equation}
and $\ketbra{x}-\id/d$ has one eigenvalue $1-1/d$ and $d-1$ eigenvalues $-1/d$,
\begin{equation}
\dtr(\rho_x^a,\rho_x^b)
=(a-b)\left(1-\frac1d\right).
\end{equation}
Averaging over $x$ matches the lower bound and proves the corollary.
\hfill$\square$

\begin{corollary}[Stochastic-branch converse]
\label{cor:stochastic-converse}
If $Z\sim\pi_t$ is input-independent and every $\cR_{t,z}$ is label-independent CPTP, then
\begin{equation}
\mathbb E_{X,Z}\dtr(\rho_{X,t-1},\cR_{t,Z}(\rho_{X,t}))
\ge\varepsilon_t^\star
\ge\max\{\Omega_{d,m}^{\leftarrow}(\Delta_t),\varepsilon_t^{\rm geo}\}.
\label{eq:stochastic-converse-main}
\end{equation}
\end{corollary}

\subsection{Proof of the stochastic-branch converse}
For every fixed latent value $z$, the branch $\cR_{t,z}$ is itself a label-independent CPTP channel and is therefore feasible in the infimum defining $\varepsilon_t^\star$. Hence
\[
 \sum_xp_x\dtr\!\left(\rho_{x,t-1},\cR_{t,z}(\rho_{x,t})\right)\geq\varepsilon_t^\star
\]
for every $z$. Averaging with respect to the input-independent latent law $\pi_t$ proves the first inequality in \cref{eq:stochastic-converse-main}; the second is the strengthened appendix bound \cref{eq:recovery-lower-bracket-app}.
\hfill$\square$

\begin{corollary}[Blockwise common-channel recovery]
\label{cor:blockwise-recovery-main}
For $0\le s<t\le T$ with $\lambda_s>0$, let $\Delta_{s:t}=\chi_s-\chi_t$ and define $\varepsilon_{s:t}^\star$ analogously. Then
\begin{equation}
\max\!\left\{\Omega_{d,m}^{\leftarrow}(\Delta_{s:t}),\,(\lambda_s-\lambda_t)\mathfrak G(\cE)\right\}
\le\varepsilon_{s:t}^\star
\le\sqrt{1-e^{-\Delta_{s:t}}}.
\label{eq:block-sandwich-main}
\end{equation}
Moreover, a common channel $\cR_{s\leftarrow t}$ satisfies
\begin{equation}
\sum_xp_x\Fr(\rho_{x,s},\cR_{s\leftarrow t}(\rho_{x,t}))\ge e^{-\Delta_{s:t}/2},
\quad
\sum_xp_x\dtr(\rho_{x,s},\cR_{s\leftarrow t}(\rho_{x,t}))\le\sqrt{1-e^{-\Delta_{s:t}}}.
\label{eq:block-fidelity-main}
\end{equation}
\end{corollary}

\subsection{Proof of the blockwise recovery corollary}
For $\lambda_s>0$, the depolarizing semigroup identity gives
\begin{equation}
  \rho_{x,t}=\cD_{\lambda_t/\lambda_s}(\rho_{x,s}),
  \qquad
  \bar\rho_t=\cD_{\lambda_t/\lambda_s}(\bar\rho_s).
\end{equation}
Apply the universal recovery theorem to the reference state $\bar\rho_s$ and the cumulative channel $\cD_{\lambda_t/\lambda_s}$. The support argument used in the one-step proof applies unchanged, so one channel works for all labels. The Holevo-continuity argument from \cref{thm:two-sided-recoverability} applies verbatim to this cumulative channel. The pairwise argument also applies to the block endpoints; depolarizing homogeneity gives pairwise distinguishability loss $(\lambda_s-\lambda_t)\dtr(\rho_x,\rho_y)$, and the same LP-scaling argument as in \cref{eq:pairwise-factor-main} contributes $(\lambda_s-\lambda_t)\mathfrak G(\cE)$. Together these give the lower side of \cref{eq:block-sandwich-main}. If
\begin{equation}
  \delta^{\rm rel}_{x;s:t}:=D(\rho_{x,s}\Vert\bar\rho_s)-D(\rho_{x,t}\Vert\bar\rho_t),
\end{equation}
then $\delta^{\rm rel}_{x;s:t}\geq0$, $\sum_xp_x\delta^{\rm rel}_{x;s:t}=\Delta_{s:t}$. Define the blockwise fidelity
\begin{equation}
  F_{x;s:t}:=\Fr\bigl(\rho_{x,s},\cR_{s\leftarrow t}(\rho_{x,t})\bigr).
\end{equation}
Universal recoverability gives $F_{x;s:t}\geq e^{-\delta^{\rm rel}_{x;s:t}/2}$. Averaging and applying Jensen to the convex function $r\mapsto e^{-r/2}$ proves \cref{eq:block-fidelity-main}. For the trace-distance statement, apply the upper Fuchs--van de Graaf inequality and then Jensen to the concave function $f\mapsto\sqrt{1-f^2}$:
\begin{align}
  \sum_xp_x\dtr\bigl(\rho_{x,s},\cR_{s\leftarrow t}(\rho_{x,t})\bigr)
  &\leq \sum_xp_x\sqrt{1-F_{x;s:t}^2}\\
  &\leq \sqrt{1-\left(\sum_xp_xF_{x;s:t}\right)^2}\\
  &\leq \sqrt{1-e^{-\Delta_{s:t}}}.
\end{align}
The same channel proves the upper side of \cref{eq:block-sandwich-main}, completing \cref{cor:blockwise-recovery-main}.
\hfill$\square$

\begin{corollary}[Clipped-loss attainability]
\label{cor:clipped-budget-attainable}
Viewing $\cR_t^{\mathrm U}$ as a degenerate stochastic kernel, for every $\gamma\in(0,1)$,
\begin{equation}
  L_{\gamma,t}(\cR_t^{\mathrm U})
  \leq\cL_t^{\log}(\cR_t^{\mathrm U})
  \leq\Delta_t.
\end{equation}
Thus the same decrement is an attainable budget for the clipped loss used in training.
\end{corollary}

\subsection{Proof of clipped-budget attainability}
For every fidelity value $F\in[0,1]$,
\begin{equation}
  -2\log\max\{\gamma,F\}\leq-2\log F,
\end{equation}
with the right-hand side interpreted as $+\infty$ at $F=0$. Apply this pointwise to the deterministic universal recovery channel and average over $x$. The second inequality is \cref{thm:recovery-budget}.
\hfill$\square$

\begin{proposition}[Diamond-norm degradation of the budget]
\label{prop:approximation}
Let $\cR_{\theta,t}$ be a deterministic CPTP channel and define
\begin{equation}
  \varepsilon_{\diamond,t}(\theta)=\half\norm{\cR_{\theta,t}-\cR_t^{\mathrm U}}_\diamond.
\end{equation}
Then
\begin{equation}
  \sum_xp_x\Fr\bigl(\rho_{x,t-1},\cR_{\theta,t}(\rho_{x,t})\bigr)
  \geq\pos{e^{-\Delta_t/2}-\sqrt{2\varepsilon_{\diamond,t}(\theta)}}.
  \label{eq:approx-fidelity}
\end{equation}
Whenever the positive part is nonzero, the cq risk is at most minus twice its logarithm. No analogous statewise average-log bound follows without an individual fidelity lower bound.
\end{proposition}

\subsection{Proof of the channel-approximation proposition}
By the definition of diamond norm,
\begin{equation}
  \dtr\bigl(\cR_{\theta,t}(\rho),\cR_t^{\mathrm U}(\rho)\bigr)
  \leq\varepsilon_{\diamond,t}(\theta)
\end{equation}
for every input state. The Powers--St\o rmer inequality and H\"older's inequality imply the standard fidelity continuity estimate
\begin{equation}
  \abs{\Fr(\sigma,\rho)-\Fr(\sigma,\rho')}
  \leq\sqrt{\norm{\rho-\rho'}_1}
  =\sqrt{2\dtr(\rho,\rho')}.
  \label{eq:fidelity-continuity}
\end{equation}
Hence every statewise fidelity is reduced by at most $\sqrt{2\varepsilon_{\diamond,t}(\theta)}$. Averaging and applying the appendix-local fidelity consequence \cref{eq:average-fidelity-app} proves \cref{eq:approx-fidelity}.
\hfill$\square$

\section{Proofs for recoverability--coverage separation}
\label{app:coverage-separation}

\subsection{Assignment blindness of product-coupled scores}
For the good model in \cref{prop:assignment-blindness}, conditioning on $Z=j$ gives the score law in \cref{eq:score-regularity-main}. Therefore
\begin{equation}
\Law\!\left(S(\rho_X,\rho_Z)\right)
=\frac1m\sum_{j=1}^m\nu_S
=\nu_S
=\Law\!\left(S(\rho_X,\rho_{j_0})\right).
\end{equation}
The good output law is $\mu$, so its Wasserstein error is zero. The collapsed output law is the Dirac measure $\delta_{\rho_{j_0}}$; its coupling with $\mu$ is unique, giving
\begin{equation}
\Wtr(\mu,\delta_{\rho_{j_0}})
=\frac1m\sum_x\dtr(\rho_x,\rho_{j_0}).
\end{equation}
This proves \cref{prop:assignment-blindness}. \hfill$\square$

\subsection{Fixed two-qubit noncommuting separation}
Four-dimensional SIC-POVMs exist and admit an explicit two-qubit realization \citep{renes2004sic,zhu2010twoqubitsic}. Let $P_j=\ketbra{\psi_j}$, $j=1,\ldots,16$, denote its normalized rank-one projectors. The SIC identities are
\begin{equation}
\frac1{16}\sum_{j=1}^{16}P_j=\frac{\id}{4},
\qquad
\Tr(P_iP_j)=\frac15\quad(i\neq j).
\label{eq:sic-identities-app}
\end{equation}
Identifying the four-dimensional Hilbert space with two qubits, define
\begin{equation}
\rho_j:=\frac23 P_j+\frac1{12}\id,
\qquad
\mu_{\rm SIC}:=\frac1{16}\sum_{j=1}^{16}\delta_{\rho_j},
\qquad
\tau:=\frac{\id}{4}.
\label{eq:sic-target-states-app}
\end{equation}
The average state is $\tau$, and complete depolarization maps every $\rho_j$ to $\tau$. Moreover, for $i\neq j$,
\begin{equation}
[\rho_i,\rho_j]=\frac49[P_i,P_j]\neq0.
\end{equation}
Indeed, two distinct rank-one projectors commute only if their ranges are orthogonal, whereas \cref{eq:sic-identities-app} gives nonzero, nonunit overlap. Thus the target ensemble is pairwise noncommuting.

Let $X,Z$ be independent and uniform on $\{1,\ldots,16\}$. Define the good and collapsed stochastic reverse models by replacer branches,
\begin{equation}
\cR_z^{\rm good}(A)=\Tr(A)\rho_z,
\qquad
\cR_z^{\rm col}(A)=\Tr(A)\rho_1.
\label{eq:sic-replacer-branches-app}
\end{equation}
The good output law is exactly $\mu_{\rm SIC}$, while the collapsed output law is $\delta_{\rho_1}$.

We first compute the off-diagonal root fidelity. More generally, for
\begin{equation}
\rho_P=v\id+\lambda P,
\qquad
\rho_Q=v\id+\lambda Q,
\qquad
u:=v+\lambda,
\qquad
s:=\Tr(PQ),
\end{equation}
with rank-one projectors $P,Q$ in dimension $d$, the orthogonal complement of $\operatorname{span}\{P,Q\}$ contributes $(d-2)v$ to the root fidelity. On the two-dimensional span, the positive matrix $\sqrt{\rho_P}\rho_Q\sqrt{\rho_P}$ has trace $2u v+\lambda^2s$ and determinant $(u v)^2$. Hence
\begin{equation}
\Fr(\rho_P,\rho_Q)
=(d-2)v+\sqrt{4u v+\lambda^2s}.
\label{eq:depolarized-pure-fidelity-app}
\end{equation}
For \cref{eq:sic-target-states-app}, $d=4$, $\lambda=2/3$, $v=1/12$, $u=3/4$, and $s=1/5$ off the diagonal. Therefore
\begin{equation}
f_{\rm SIC}
:=\Fr(\rho_i,\rho_j)
=\frac16+\sqrt{\frac{61}{180}}
=\frac{5+\sqrt{305}}{30}
\qquad(i\neq j).
\label{eq:sic-offdiag-fidelity-app}
\end{equation}
The fidelity equals one on the diagonal. In the good model $X=Z$ with probability $1/16$; in the collapsed model $X=1$ with the same probability. Thus both complete scalar fidelity laws are
\begin{equation}
F=\begin{cases}
1,&\text{with probability }1/16,\\
f_{\rm SIC},&\text{with probability }15/16,
\end{cases}
\label{eq:sic-fidelity-law-app}
\end{equation}
proving \cref{eq:fidelity-law-separation-main}. Their common unclipped local log-fidelity risk is
\begin{equation}
L_{\rm SIC}^{\log}
=-\frac{15}{8}\log f_{\rm SIC}.
\label{eq:sic-logrisk-app}
\end{equation}

Each $\rho_j$ has spectrum $\{3/4,1/12,1/12,1/12\}$. Since the average state is $\tau$, the complete-depolarization cq-information decrement is
\begin{equation}
\Delta_{\rm SIC}
=\log4-S(\rho_j)
=\frac12\log3.
\label{eq:sic-budget-app}
\end{equation}
The budget inequality is strict. To verify it without a floating-point comparison, note first that $\sqrt{305}>87/5$, hence
\begin{equation}
f_{\rm SIC}>\frac{56}{75}.
\end{equation}
Moreover,
\begin{equation}
\left(\frac{56}{75}\right)^3>\frac{27}{65}
\quad\text{and}\quad
\left(\frac{27}{65}\right)^5>\frac1{81};
\end{equation}
the two cross-multiplied numerator differences are respectively $24415$ and $1970842$. Consequently $f_{\rm SIC}^{15}>1/81$, which is equivalent to
\begin{equation}
-\frac{15}{8}\log f_{\rm SIC}<\frac12\log3=\Delta_{\rm SIC}.
\label{eq:sic-budget-verified-app}
\end{equation}
This proves the exact budget-feasibility claim in \cref{eq:twoq-budget-separation-main}.

Finally, for distinct SIC projectors,
\begin{equation}
\dtr(P_i,P_j)=\sqrt{1-\Tr(P_iP_j)}=\frac2{\sqrt5},
\end{equation}
and homogeneity under the common depolarizing mixture gives
\begin{equation}
\dtr(\rho_i,\rho_j)=\frac{4}{3\sqrt5}
\qquad(i\neq j).
\end{equation}
The good model has zero endpoint error. For the collapsed model the second Wasserstein marginal is a Dirac measure, so the coupling is unique and
\begin{align}
\Wtr(\mu_{\rm SIC},\delta_{\rho_1})
&=\frac1{16}\sum_{j=1}^{16}\dtr(\rho_j,\rho_1)\\
&=\frac{15}{16}\frac{4}{3\sqrt5}
=\frac{\sqrt5}{4}>\frac12.
\end{align}
This proves \cref{thm:coverage-separation-main}. \hfill$\square$

\subsection{Amplification to an arbitrary constant below one}
We prove \cref{cor:coverage-separation-amplified} by the symmetric commuting family used to amplify the separation. Fix $d\geq2$ and $\lambda\in(0,1)$, let $\tau=\id/d$, and define
\begin{equation}
\rho_x^\lambda=\lambda\ketbra{x}+(1-\lambda)\tau,
\qquad
\mu_\lambda^{(d)}=\frac1d\sum_{x=1}^{d}\delta_{\rho_x^\lambda}.
\end{equation}
The complete-depolarization channel maps every $\rho_x^\lambda$ to $\tau$, so the information decrement of this one-step forward problem is
\begin{equation}
\Delta_d(\lambda)=\chi(\cE_\lambda^{(d)}).
\label{eq:separation-delta-app}
\end{equation}

Let $Z$ be uniform on $\{1,\ldots,d\}$ and independent of the target label $X$. Define the good branch family and the collapsed branch family by the CPTP replacer channels
\begin{align}
\cR_z^{\rm good}(A)&=\Tr(A)\rho_z^\lambda,\\
\cR_z^{\rm col}(A)&=\Tr(A)\rho_1^\lambda
\qquad\text{for every }z.
\end{align}
The good output law is exactly $\mu_\lambda^{(d)}$, whereas the collapsed output law is $\delta_{\rho_1^\lambda}$.

Set
\begin{equation}
u_\lambda:=\frac{1+(d-1)\lambda}{d},
\qquad
v_\lambda:=\frac{1-\lambda}{d}.
\end{equation}
The states commute in the computational basis. For $x\neq z$, their root fidelity is
\begin{equation}
f_d(\lambda)
=2\sqrt{u_\lambda v_\lambda}+(d-2)v_\lambda,
\label{eq:offdiag-fidelity-app}
\end{equation}
while the fidelity is one when $x=z$. In the good model, $X=Z$ with probability $1/d$; in the collapsed model, $X=1$ with probability $1/d$. Hence both scalar fidelity random variables have the same law,
\begin{equation}
F=\begin{cases}
1,&\text{with probability }1/d,\\
f_d(\lambda),&\text{with probability }(d-1)/d,
\end{cases}
\end{equation}
and their common unclipped branch-averaged log risk is
\begin{equation}
\cL_d^{\log,\rm br}(\lambda)
=-2\frac{d-1}{d}\log f_d(\lambda).
\label{eq:separation-logrisk-app}
\end{equation}

The good model reproduces $\mu_\lambda^{(d)}$, so its Wasserstein error is zero. For the collapsed model the second marginal is a Dirac measure, hence the coupling is unique and
\begin{align}
\Wtr\bigl(\mu_\lambda^{(d)},\delta_{\rho_1^\lambda}\bigr)
&=\frac1d\sum_x\dtr(\rho_x^\lambda,\rho_1^\lambda)\\
&=\lambda\left(1-\frac1d\right),
\label{eq:collapsed-wasserstein-app}
\end{align}
because $\dtr(\rho_x^\lambda,\rho_y^\lambda)=\lambda$ for $x\neq y$.

The average state of $\cE_\lambda^{(d)}$ is $\tau$, so
\begin{equation}
\Delta_d(\lambda)
=\log d-S(\rho_x^\lambda)
=\log d+\nu_\lambda\log \nu_\lambda+(d-1)v_\lambda\log v_\lambda.
\label{eq:chi-orthogonal-app}
\end{equation}
For every fixed $\lambda\in(0,1)$, as $d\to\infty$,
\begin{equation}
\Delta_d(\lambda)
=\lambda\log d-h_2(\lambda)+o(1),
\end{equation}
whereas \cref{eq:offdiag-fidelity-app} gives
\begin{equation}
f_d(\lambda)\longrightarrow1-\lambda,
\qquad
\cL_d^{\log,\rm br}(\lambda)
\longrightarrow-2\log(1-\lambda)<\infty.
\end{equation}
Thus for every fixed $\lambda\in(0,1)$ there is a finite $d_0(\lambda)$ such that
\begin{equation}
\cL_d^{\log,\rm br}(\lambda)\leq\Delta_d(\lambda)
\qquad\text{for all }d\geq d_0(\lambda).
\end{equation}
Given $c<1$, choose $\lambda\in(c,1)$. Equation~\eqref{eq:collapsed-wasserstein-app} exceeds $c$ for all sufficiently large $d$, while the budget inequality also holds for sufficiently large $d$. We may choose such a $d$ along the subsequence $d=2^n$, proving the corollary for qubit Hilbert spaces. \hfill$\square$

\paragraph{Small-noise complement.}
The amplification family is not purely a large-budget construction. For fixed $d>2$ and small $\lambda$, direct Taylor expansion gives
\begin{align}
\chi(\cE_\lambda^{(d)})
&=\frac{d-1}{2}\lambda^2-
\frac{(d-1)(d-2)}{6}\lambda^3+O(\lambda^4),\\
\cL_d^{\log,\rm br}(\lambda)
&=\frac{d-1}{2}\lambda^2-
\frac{(d-1)(d-2)}{4}\lambda^3+O(\lambda^4).
\end{align}
Hence
\begin{equation}
\chi(\cE_\lambda^{(d)})-\cL_d^{\log,\rm br}(\lambda)
=\frac{(d-1)(d-2)}{12}\lambda^3+O(\lambda^4)>0
\end{equation}
for all sufficiently small positive $\lambda$.

\section{Petz and averaged rotated-Petz maps}
\label{app:petz}
The recovery theorem is used as an existence result, but an explicit representative is useful for teacher construction. Let $\sigma$ be a state and $\cN$ a channel. The Hilbert--Schmidt adjoint $\cN^\dagger$ is defined by $\Tr[Y\cN(X)]=\Tr[\cN^\dagger(Y)X]$. On $\supp\cN(\sigma)$, the Petz map is \citep{petz1986sufficient,fawzi2015approximate,junge2018universal}
\begin{equation}
  \cP_{\sigma,\cN}(A)
  =\sigma^{1/2}\cN^\dagger\!\left(
  \cN(\sigma)^{-1/2}A\cN(\sigma)^{-1/2}
  \right)\sigma^{1/2},
  \label{eq:petz}
\end{equation}
where inverse powers are generalized inverses on their supports. Following the convention of \citet{junge2018universal}, define the rotated map
\begin{equation}
  \cR^{s}_{\sigma,\cN}(A)
  =\sigma^{-is}\cP_{\sigma,\cN}\!\left(
  \cN(\sigma)^{is}A\cN(\sigma)^{-is}
  \right)\sigma^{is}.
  \label{eq:rotated-petz}
\end{equation}
Each map is CPTP on $\supp\cN(\sigma)$. Let $\Pi$ project onto this support and choose a state $\tau_0$ supported on $\supp\sigma$. The completion
\begin{equation}
  \widetilde\cT(A)=\cT(\Pi A\Pi)+\Tr[(\id-\Pi)A]\tau_0
  \label{eq:recovery-completion}
\end{equation}
is CPTP on the full output space and agrees with $\cT$ on every relevant output $\cN(\rho)$ with $\supp\rho\subseteq\supp\sigma$. Indeed, $\rho\leq c\sigma$ for some finite $c$, so $\cN(\rho)\leq c\cN(\sigma)$ and the output support is contained in $\supp\cN(\sigma)$.

A universal recovery channel satisfying the inequality used in \cref{eq:universal-recovery} is
\begin{equation}
  \cR^{\mathrm U}_{\sigma,\cN}
  =\int_{\mathbb R}\beta_0(s)\cR^{s/2}_{\sigma,\cN}\,\dd s,
  \qquad
  \beta_0(s)=\frac{\pi}{2\,[\cosh(\pi s)+1]}.
  \label{eq:averaged-rotated-petz}
\end{equation}
The density $\beta_0$ integrates to one. The ordinary Petz map is the zero-rotation member of this family; it is not assumed to satisfy the universal fidelity remainder by itself for every ensemble.

\section{Uniform finite-sample calibration}
\label{app:generalization}
For the covering-number proof only, write
\begin{equation}
\ell_{\gamma,t}(\theta;x,z):=-2\log\max\{\gamma,\Fr(\rho_{x,t-1},\cR_{\theta,t,z}(\rho_{x,t}))\}.
\label{eq:sample-clipped-loss-main}
\end{equation}
This section supplies the technical details behind \cref{thm:generalization}.  Assume the branch unitary has $P$ trainable rotations,
\begin{equation}
  U_\theta(u,h,z)=W_P(u,h,z)e^{-i\theta_PH_P/2}\cdots
  W_1(u,h,z)e^{-i\theta_1H_1/2}W_0(u,h,z),
  \label{eq:pqc-main}
\end{equation}
with $\theta\in[-\pi,\pi]^P$, $\norm{H_j}_\infty\leq1$, and fixed conditional unitaries $W_j(u,h,z)$. The theorem therefore certifies trainable parameters entering through the displayed bounded-generator rotations; the time/latent conditioning inside $W_j$ is a fixed function of $(u,h,z)$. A separately trainable classical conditioner would require its own Lipschitz and covering-number control and is not included in the bound below. The resulting parameter dependence is consistent with general QML bounds based on trainable gates \citep{caro2022generalization,cai2022sample}.

For the concrete net tolerance $\varepsilon=2\log(1/\gamma)/\sqrt n$, define
\begin{equation}
  \varepsilon_{\rm gen}
  :=\frac{4\log(1/\gamma)}{\sqrt n}
  +2\log(1/\gamma)\sqrt{
  \frac{P\log\!\left(1+\frac{2\pi Pn}{\gamma^2\log^2(1/\gamma)}\right)
  +\log(T/\delta)}{2n}}.
  \label{eq:epsilon-gen-app}
\end{equation}
For a trained parameter $\widehat\theta$, the main-text calibration slack may be taken as
\begin{equation}
\varepsilon_{\rm cal}:=\eta_{\rm train}+\varepsilon_{\rm gen}+\max_t|\widehat\Delta_t-\Delta_t|.
\label{eq:epsilon-cal-app}
\end{equation}
The proof below keeps a generic net tolerance $\varepsilon$ until the final substitution.

A telescoping expansion of \cref{eq:pqc-main} gives
\begin{equation}
  \norm{U_\theta-U_{\theta'}}_\infty
  \leq\half\norm{\theta-\theta'}_1
  \leq\frac{\sqrt P}{2}\norm{\theta-\theta'}_2.
  \label{eq:unitary-parameter}
\end{equation}
Unitary conjugation followed by partial trace implies
\begin{equation}
  \dtr\bigl(\cR_{\theta,t,z}(\rho),\cR_{\theta',t,z}(\rho)\bigr)
  \leq\norm{U_\theta(u,h,z)-U_{\theta'}(u,h,z)}_\infty.
\end{equation}
Combining this with \cref{eq:fidelity-continuity} yields
\begin{equation}
  \abs{F_\theta-F_{\theta'}}
  \leq P^{1/4}\norm{\theta-\theta'}_2^{1/2}.
\end{equation}
The scalar function $f\mapsto-2\log\max\{\gamma,f\}$ is $2/\gamma$-Lipschitz on $[0,1]$. Therefore
\begin{equation}
  \abs{\ell_{\gamma,t}(\theta)-\ell_{\gamma,t}(\theta')}
  \leq\frac{2P^{1/4}}{\gamma}\norm{\theta-\theta'}_2^{1/2}.
  \label{eq:loss-holder}
\end{equation}

Choose a Euclidean parameter-net radius
\begin{equation}
  r=\frac{\gamma^2\varepsilon^2}{4\sqrt P}.
\end{equation}
Then \cref{eq:loss-holder} makes the induced loss functions uniformly $\varepsilon$-close. The parameter cube is contained in the Euclidean ball of radius $\pi\sqrt P$, whose $r$-covering number is at most
\begin{equation}
  N_r\leq\left(1+\frac{2\pi\sqrt P}{r}\right)^P
  =\left(1+\frac{8\pi P}{\gamma^2\varepsilon^2}\right)^P.
  \label{eq:covering-number}
\end{equation}
For every fixed net point and timestep, Hoeffding's inequality applies because $0\leq\ell_{\gamma,t}\leq 2\log(1/\gamma)$. A union bound over $T N_r$ pairs shows that, with probability at least $1-\delta$, every net loss satisfies
\begin{equation}
  L_{\gamma,t}(\theta_j)
  \leq\widehat L_{\gamma,t}(\theta_j)
  +2\log(1/\gamma)\sqrt{\frac{\log(TN_r/\delta)}{2n}}.
\end{equation}
For an arbitrary $\theta$, choose a net point $\theta_j$ within radius $r$. The population and empirical risks each change by at most $\varepsilon$, adding $2\varepsilon$. Substituting \cref{eq:covering-number} gives the corresponding bound for a generic $\varepsilon$; setting $\varepsilon=2\log(1/\gamma)/\sqrt n$ yields \cref{eq:epsilon-gen-app} and proves \cref{thm:generalization}.
\hfill$\square$

\subsection{Proof of population calibration}
For every $t$,
\begin{align}
  L_{\gamma,t}(\widehat\theta)-\Delta_t
  \leq{}&\widehat L_{\gamma,t}(\widehat\theta)-\widehat\Delta_t
  +\varepsilon_{\rm gen}
  +\abs{\widehat\Delta_t-\Delta_t}.
\end{align}
Taking positive parts and the maximum over $t$ proves the result.
\hfill$\square$

Define the fidelity lower-bound function
\begin{equation}
  \underline F_\gamma(a)
  :=\pos{e^{-a/2}-\frac{\gamma a}{2\log(1/\gamma)}}.
  \label{eq:fidelity-lower-bound-function}
\end{equation}
\begin{lemma}[Clipped risk implies average fidelity]
\label{lem:clipped-fidelity}
If $L_{\gamma,t}(\theta)\leq a$, then, for $(X,Z)\sim p\times\pi_t$,
\begin{equation}
  \mathbb E_{X,Z}\Fr\!\left(
  \rho_{X,t-1},\cR_{\theta,t,Z}(\rho_{X,t})
  \right)
  \geq \underline F_\gamma(a).
  \label{eq:clipped-to-fidelity}
\end{equation}
\end{lemma}

\subsection{Proof of the clipped-fidelity lemma}
Let $F=\Fr(\rho_{X,t-1},\cR_{\theta,t,Z}(\rho_{X,t}))$ under $(X,Z)\sim p\times\pi_t$, and let $M_\gamma(F)=\max\{\gamma,F\}$. If $L_{\gamma,t}\leq a$, then
\begin{equation}
  \mathbb E\log M_\gamma(F)=-\half L_{\gamma,t}\geq-a/2.
\end{equation}
By concavity of $\log$,
\begin{equation}
  \mathbb E M_\gamma(F)\geq e^{-a/2}.
\end{equation}
On the event $\{F<\gamma\}$, the clipped loss equals $2\log(1/\gamma)$, hence
\begin{equation}
  \Pr(F<\gamma)\leq\frac{a}{2\log(1/\gamma)}.
\end{equation}
Finally,
\begin{align}
  \mathbb EF
  &=\mathbb EM_\gamma(F)-\mathbb E(\gamma-F)_+\\
  &\geq e^{-a/2}-\gamma\Pr(F<\gamma)\\
  &\geq e^{-a/2}-\frac{\gamma a}{2\log(1/\gamma)}.
\end{align}
Since fidelity is nonnegative, taking the positive part proves \cref{eq:clipped-to-fidelity}.
\hfill$\square$

\section{Local-to-global distributional bounds}
\label{app:endpoint}
The main text states the nonexpansive bound directly in Wasserstein distance. For the refined proof, define
\begin{equation}
D_t:=\Wtr(\mu_t,\nu_t),
\qquad
\kappa(K):=\sup_{\rho\ne\sigma}
\frac{\Wtr(K(\rho,\cdot),K(\sigma,\cdot))}{\dtr(\rho,\sigma)},
\end{equation}
and write $\kappa_t=\kappa(K_{\widehat\theta,t})$. The contraction-weighted form is
\begin{equation}
D_0\le
\left(\prod_{j=1}^T\kappa_j\right)D_T+
\sum_{t=1}^T\left(\prod_{j=1}^{t-1}\kappa_j\right)
\psi_\gamma(\Delta_t+\varepsilon_{\rm cal}).
\label{eq:endpoint-certificate-main}
\end{equation}
For the input-independent CPTP branch kernels considered in the main text, $\kappa_t\le1$, and this reduces to \cref{eq:endpoint-simple-main}.
The preceding results are local in time. Here they are lifted to probability measures on quantum-state space and composed along the learned reverse chain.

Let $\pi_t$ be a probability measure on a standard Borel latent space $\mathcal Z$.  The state-valued Markov kernel used in the main text is, for every Borel set $B\subseteq\cS(Q)$,
\begin{equation}
  K_{\theta,t}(\rho,B)
  =\int_{\mathcal Z}\mathbf 1_B\!\left(\cR_{\theta,t,z}(\rho)\right)\pi_t(\dd z),
  \label{eq:kernel-induced-by-channel-main}
\end{equation}
which is equivalent to the law notation in \cref{eq:kernel-law-main}. Using the fidelity lower bound in \cref{eq:fidelity-lower-bound-function}, define
\begin{equation}
  \psi_\gamma(a)=\sqrt{1-\underline F_\gamma(a)^2}.
  \label{eq:q-psi-app}
\end{equation}
This is exactly the main-text transfer function in \cref{eq:psi-main}.

\begin{proposition}[Local distributional recovery]
\label{prop:local-distribution}
If $L_{\gamma,t}(\theta)\leq a_t$, then
\begin{equation}
  \Wtr(\mu_{t-1},K_{\theta,t\#}\mu_t)
  \leq\psi_\gamma(a_t).
  \label{eq:local-distribution-bound}
\end{equation}
Without any additional positivity assumption, the unclipped average fidelity also gives the sharper bound
\begin{equation}
  \Wtr(\mu_{t-1},K_{\theta,t\#}\mu_t)
  \leq
  \sqrt{1-
  \left(
  \mathbb E_{X,Z}\Fr\!\left(
  \rho_{X,t-1},\cR_{\theta,t,Z}(\rho_{X,t})
  \right)
  \right)^2}.
  \label{eq:unclipped-local-distribution-bound}
\end{equation}
\end{proposition}

\subsection{Proof of the local distributional bound}
Construct a coupling by sampling $X\sim p$, setting $\rho=\rho_{X,t-1}$ and $\rho_t=\cN_t(\rho)$, and sampling $\widehat\rho\sim K_{\theta,t}(\rho_t,\cdot)$. Its marginals are $\mu_{t-1}$ and $K_{\theta,t\#}\mu_t$, so
\begin{equation}
  \Wtr(\mu_{t-1},K_{\theta,t\#}\mu_t)
  \leq\mathbb E\dtr(\rho,\widehat\rho).
\end{equation}
By \cref{eq:fvdg-main},
\begin{equation}
  \dtr(\rho,\widehat\rho)
  \leq\sqrt{1-\Fr(\rho,\widehat\rho)^2}.
\end{equation}
The function $g(f)=\sqrt{1-f^2}$ is concave on $[0,1]$, so Jensen's inequality and \cref{lem:clipped-fidelity} give
\begin{equation}
  \mathbb E\dtr(\rho,\widehat\rho)
  \leq\sqrt{1-(\mathbb E\Fr)^2}
  \leq\psi_\gamma(a_t).
\end{equation}
The unclipped statement is the same Jensen bound applied directly to the positive average fidelity, without the clipping relaxation.
\hfill$\square$

\subsection{Proof of latent-channel nonexpansiveness}
For two inputs $\rho,\sigma$, couple their outputs using the same latent variable $Z$. Contractivity of trace distance under each CPTP channel gives
\begin{align}
  \Wtr(K_{\theta,t}(\rho,\cdot),K_{\theta,t}(\sigma,\cdot))
  &\leq\mathbb E_Z\dtr(\cR_{\theta,t,Z}(\rho),\cR_{\theta,t,Z}(\sigma))\\
  &\leq\dtr(\rho,\sigma).
\end{align}
Taking the supremum proves $\kappa(K_{\theta,t})\leq1$; in particular, $\kappa_t\leq1$ for the trained kernel.
\hfill$\square$

\begin{lemma}[Dobrushin contraction]
\label{lem:dobrushin-pushforward}
For any Borel Markov kernel $K$ on $(\cS(Q),\dtr)$ with Dobrushin coefficient $\kappa$,
\begin{equation}
  \Wtr(K_\#\mu,K_\#\nu)\leq\kappa\Wtr(\mu,\nu)
\end{equation}
for all probability measures $\mu,\nu$ on $\cS(Q)$.
\end{lemma}

\subsection{Proof of the Dobrushin pushforward lemma}
Because $\cS(Q)$ is compact, Kantorovich--Rubinstein duality applies. For any real-valued 1-Lipschitz function $f$, define $(Kf)(\rho)=\int f(\sigma)K(\rho,\dd\sigma)$. The definition of $\kappa$ and the dual formula for $W_1$ imply
\begin{equation}
  |(Kf)(\rho)-(Kf)(\rho')|
  \leq\Wtr(K(\rho,\cdot),K(\rho',\cdot))
  \leq\kappa\dtr(\rho,\rho'),
\end{equation}
so $Kf$ is $\kappa$-Lipschitz. Applying duality once more,
\begin{align}
  \Wtr(K_\#\mu,K_\#\nu)
  &=\sup_{\operatorname{Lip}(f)\leq1}
    \left[\int Kf\,\dd\mu-\int Kf\,\dd\nu\right]\\
  &\leq\kappa\Wtr(\mu,\nu).
\end{align}
\hfill$\square$

\subsection{Proof of the compositional endpoint bound}
The triangle inequality and \cref{lem:dobrushin-pushforward} give
\begin{align}
  D_{t-1}
  &\leq\Wtr(\mu_{t-1},K_{\theta,t\#}\mu_t)
  +\Wtr(K_{\theta,t\#}\mu_t,K_{\theta,t\#}\nu_t)\\
  &\leq\psi_\gamma(\Delta_t+\varepsilon_{\rm cal})+\kappa_tD_t,
\end{align}
where \cref{eq:population-calibration-main} provides $L_{\gamma,t}\leq\Delta_t+\varepsilon_{\rm cal}$. Iterating the recursion from $t=T$ down to $t=1$ proves \cref{eq:endpoint-certificate-main}.
\hfill$\square$

\subsection{Proof of worst-step schedule optimality}
If $\kappa_t\leq1$, then \cref{eq:endpoint-certificate-main} implies
\begin{equation}
  D_0\leq D_T+\sum_{t=1}^T\psi_\gamma(\Delta_t+\varepsilon_{\rm cal})
  \leq D_T+T\psi_\gamma(\max_t\Delta_t+\varepsilon_{\rm cal}),
\end{equation}
because $\psi_\gamma$ is nondecreasing. Holding $D_T$ and $\varepsilon_{\rm cal}$ fixed across candidate grids, \cref{thm:forward-clock} shows that the equal-information grid uniquely minimizes $\max_t\Delta_t$, so it minimizes the displayed bound. If $\underline F_\gamma>0$ on the relevant range, then $\underline F_\gamma$ is strictly decreasing and $\psi_\gamma$ is strictly increasing, which transfers uniqueness to the bound. If $\underline F_\gamma=0$, $\psi_\gamma=1$ and ties can occur.
\hfill$\square$

\section{Terminal randomness and approximate clocks}
The endpoint theorem allows a general terminal prior. The first proposition isolates the role of stochasticity at complete depolarization; the second relates an estimated information grid to the true decrements on that selected grid.

\begin{proposition}[Complete depolarization and stochastic generation]
\label{prop:terminal-randomness}
If $\mu_T=\delta_{\id/d}$ and every reverse step is deterministic, then every generated measure is a Dirac measure. Conversely, any finite ensemble $\sum_zq_z\delta_{\sigma_z}$ can be generated from $\delta_{\id/d}$ by sampling $Z\sim q$ and applying the replacer channel $\cR_z(A)=\Tr(A)\sigma_z$.
\end{proposition}

\subsection{Proof of the terminal-randomness proposition}
A deterministic channel maps a Dirac measure to a Dirac measure, and compositions preserve this property. Conversely, every map $\cR_z(A)=\Tr(A)\sigma_z$ is CPTP. Sampling $Z\sim q$ therefore produces the output law $\sum_zq_z\delta_{\sigma_z}$ from any fixed input state, including $\id/d$.
\hfill$\square$

\subsection{Proof of the truncated-terminal mismatch proposition}
Set $\tau=\id/d$, the maximally mixed terminal prior used in the main text. Because the second marginal is a Dirac measure, its coupling with $\mu_T$ is unique. Hence
\begin{align}
  \Wtr(\mu_T,\delta_\tau)
  &=\sum_xp_x\dtr(\cD_{\lambda_T}(\rho_x),\tau)\\
  &=\lambda_T\sum_xp_x\dtr(\rho_x,\tau),
\end{align}
where the second equality uses $\tau=\id/d$ and
$\cD_{\lambda_T}(\rho_x)-\tau=\lambda_T(\rho_x-\tau)$. The trace distance from the maximally mixed state is convex in the state and therefore maximized on a pure state. For a pure state, $\rho-\tau$ has eigenvalues $1-1/d$ and $-1/d$ with multiplicity $d-1$, so $\dtr(\rho,\tau)\leq1-1/d$. This proves \cref{eq:terminal-mismatch-main}.
\hfill$\square$

\begin{proposition}[Approximate information clock]
\label{prop:clock-estimation}
Suppose $\sup_{\lambda\in[\lambda_T,1]}|\widehat\chi(\lambda)-\chi(\lambda)|\leq\varepsilon_\chi$. Assume the estimated endpoint drop $\widehat\Delta_{\rm tot}:=\widehat\chi(1)-\widehat\chi(\lambda_T)$ is positive and that an exact equalizing grid exists; let $\widehat\lambda_t$ denote such a grid between the fixed endpoints. Define
\begin{equation}
  \Delta_t^{\mathrm{true}}
  :=\chi(\widehat\lambda_{t-1})-\chi(\widehat\lambda_t),
  \qquad
  \widehat\Delta_t
  :=\frac{\widehat\chi(1)-\widehat\chi(\lambda_T)}{T}.
\end{equation}
Then
\begin{align}
  \max_t|\Delta_t^{\mathrm{true}}-\widehat\Delta_t|
  &\leq2\varepsilon_\chi,
  \label{eq:clock-budget-estimation-error}\\
  \max_t|\Delta_t^{\mathrm{true}}-\Delta_{\rm tot}/T|
  &\leq2\varepsilon_\chi\left(1+\frac1T\right).
  \label{eq:clock-budget-error}
\end{align}
Thus under exact inversion the training and endpoint results may take $\varepsilon_{\rm clk}=2\varepsilon_\chi$, with $\Delta_t$ interpreted as the true decrement on the selected pilot grid.
\end{proposition}

\subsection{Proof of the clock-estimation proposition}
Let $\widehat\Delta_{\rm tot}=\widehat\chi(1)-\widehat\chi(\lambda_T)$. Uniform estimation error implies
\begin{equation}
  \abs{\widehat\Delta_{\rm tot}-\Delta_{\rm tot}}\leq2\varepsilon_\chi.
\end{equation}
Let $\widehat\lambda_t$ denote the estimated grid and set $\Delta_t^{\mathrm{true}}:=\chi(\widehat\lambda_{t-1})-\chi(\widehat\lambda_t)$. The estimated decrement is $\widehat\Delta_t=\widehat\Delta_{\rm tot}/T$, while each endpoint evaluation differs from the true value by at most $\varepsilon_\chi$. Hence
\begin{equation}
  \abs{\Delta_t^{\mathrm{true}}-\widehat\Delta_t}\leq2\varepsilon_\chi,
\end{equation}
which proves \cref{eq:clock-budget-estimation-error}. The triangle inequality gives $\abs{\Delta_t^{\mathrm{true}}-\Delta_{\rm tot}/T}\leq\abs{\Delta_t^{\mathrm{true}}-\widehat\Delta_t}+\abs{\widehat\Delta_{\rm tot}-\Delta_{\rm tot}}/T\leq2\varepsilon_\chi(1+1/T)$, proving \cref{eq:clock-budget-error}.
\hfill$\square$

\begin{theorem}[Estimated clock with inexact numerical inversion]
\label{thm:inexact-clock}
Assume the uniform event of \cref{prop:clock-estimation}. Let
\begin{equation}
  \widehat y_t=(1-t/T)\widehat\chi(1)+(t/T)\widehat\chi(\lambda_T),
\end{equation}
and suppose the numerical solver returns $\widetilde\lambda_t$ satisfying
\begin{equation}
  \abs{\widehat\chi(\widetilde\lambda_t)-\widehat y_t}\leq\varepsilon_{\rm inv}
  \qquad (t=0,\ldots,T),
  \label{eq:inexact-residual}
\end{equation}
with the fixed endpoints enforced exactly. Define
\begin{equation}
  \widetilde\Delta_t^{\mathrm{true}}
  :=\chi(\widetilde\lambda_{t-1})-\chi(\widetilde\lambda_t),
  \qquad
  \widehat\Delta_t:=\frac{\widehat\chi(1)-\widehat\chi(\lambda_T)}{T}.
\end{equation}
Then
\begin{align}
  \max_t\abs{\widetilde\Delta_t^{\mathrm{true}}-\widehat\Delta_t}
  &\leq 2\varepsilon_\chi+2\varepsilon_{\rm inv},
  \label{eq:inexact-pilot-budget}\\
  \max_t\abs{\widetilde\Delta_t^{\mathrm{true}}-\Delta_{\rm tot}/T}
  &\leq2\varepsilon_\chi\left(1+\frac1T\right)+2\varepsilon_{\rm inv}.
  \label{eq:inexact-ideal-budget}
\end{align}
Let $\lambda_t^\star$ be the exact true-information grid point satisfying
$\chi(\lambda_t^\star)=(1-t/T)\chi(1)+(t/T)\chi(\lambda_T)$.
If $\lambda_t^\star$ and $\widetilde\lambda_t$ lie in an interval on which $\chi'(\lambda)\geq m_*>0$, then
$|\widetilde\lambda_t-\lambda_t^\star|\leq(2\varepsilon_\chi+\varepsilon_{\rm inv})/m_*$ for every $t$.
\end{theorem}

\begin{proof}
The residual and uniform curve error imply
\begin{equation}
  \abs{\chi(\widetilde\lambda_t)-\widehat y_t}
  \leq\varepsilon_\chi+\varepsilon_{\rm inv}.
\end{equation}
Since $\widehat y_{t-1}-\widehat y_t=\widehat\Delta_t$, subtracting two adjacent relations proves \cref{eq:inexact-pilot-budget}. Moreover,
$\abs{\widehat\Delta_t-\Delta_{\rm tot}/T}\leq2\varepsilon_\chi/T$, which gives \cref{eq:inexact-ideal-budget}. Finally, with $y_t^\star:=(1-t/T)\chi(1)+(t/T)\chi(\lambda_T)$, the endpoint interpolation gives $|\widehat y_t-y_t^\star|\leq\varepsilon_\chi$. Combining this with $|\chi(\widetilde\lambda_t)-\widehat y_t|\leq\varepsilon_\chi+\varepsilon_{\rm inv}$ yields $|\chi(\widetilde\lambda_t)-\chi(\lambda_t^\star)|\leq2\varepsilon_\chi+\varepsilon_{\rm inv}$. The mean-value theorem on the stated interval then gives the coordinate bound.
\end{proof}

\section{Finite-shot calibration}
\label{app:finite-shot}
The main learning theorem assumes exact loss evaluation. The results in this section provide an independent hardware-facing route based on pilot tomography and held-out validation.

Recall the monotone Audenaert entropy-continuity modulus $\mathfrak f_d$ defined in \cref{eq:continuity-functions-main}. We reuse it here to control the information-clock error from state tomography.

\begin{theorem}[Finite-shot, path-uniform information clock]
\label{thm:finite-shot-clock}
Assume the probabilities $p_x$ are known. Let $\mathsf M=\{M_y\}_{y\in\mathcal Y}$ be informationally complete with Hermitian dual frame $\{F_y\}_{y\in\mathcal Y}$ and $\Lambda_{\mathsf M}=\sum_{y\in\mathcal Y}\norm{F_y}_1$. Measure $N_\chi$ copies of each $\rho_x$, let $\widehat q_{xy}$ be the empirical frequencies, and define
\begin{equation}
  \widetilde\rho_x=\sum_{y\in\mathcal Y}\widehat q_{xy}F_y,
  \qquad
  \widehat\rho_x\in\arg\min_{\sigma\in\cS(Q)}
  \norm{\sigma-\widetilde\rho_x}_1.
  \label{eq:physical-tomography-estimator}
\end{equation}
Propagate the physical estimates through the known forward channel,
\begin{equation}
  \widehat\rho_x(\lambda)=\cD_\lambda(\widehat\rho_x),
  \qquad
  \widehat\chi(\lambda)=S\left(\sum_xp_x\widehat\rho_x(\lambda)\right)
  -\sum_xp_xS(\widehat\rho_x(\lambda)).
  \label{eq:estimated-information-curve}
\end{equation}
For $\delta_\chi\in(0,1)$, set
\begin{equation}
  \varepsilon_\rho=\min\!\left\{1,
  \Lambda_{\mathsf M}\sqrt{\frac{\log(2m|\mathcal Y|/\delta_\chi)}{2N_\chi}}
  \right\}.
\end{equation}
With probability at least $1-\delta_\chi$,
\begin{equation}
  \max_x\dtr(\widehat\rho_x,\rho_x)\leq\varepsilon_\rho
\end{equation}
and, simultaneously for all $\lambda\in[0,1]$,
\begin{equation}
  |\widehat\chi(\lambda)-\chi(\lambda)|
  \leq2\mathfrak f_d(\varepsilon_\rho).
  \label{eq:uniform-chi-shot-bound}
\end{equation}
If the estimated endpoint drop $\widehat\Delta_{\rm tot}:=\widehat\chi(1)-\widehat\chi(\lambda_T)$ is positive, the estimated equal-information grid exists uniquely and its budget error obeys
\begin{equation}
  \max_t|\widehat\Delta_t-\Delta_t^{\mathrm{true}}|
  \leq4\mathfrak f_d(\varepsilon_\rho).
  \label{eq:clock-shot-exact-app}
\end{equation}
A sufficient condition on the same event is $\Delta_{\rm tot}>4\mathfrak f_d(\varepsilon_\rho)$.
\end{theorem}

\begin{theorem}[Independent finite-shot validation of recovery risk]
\label{thm:finite-shot-validation}
Condition on a trained parameter $\widehat\theta$ and on a pilot grid, both fixed before validation. At each timestep, draw $n_{\mathrm v}$ iid validation observations; the same indexed trajectory may be reused across timesteps, so only iid sampling across the validation index is required. For every target/output pair, assume that both density operators can be prepared repeatedly. For a stochastic reverse step this means that the recorded latent value or measurement branch can be conditioned on during validation. Tomograph each state from $N_{\mathrm v}$ copies using the same POVM $\mathsf M$ and the physical estimator in \cref{eq:physical-tomography-estimator}. Let $\widetilde L_{\gamma,t}(\widehat\theta)$ be the empirical clipped loss computed from the tomographic estimates. Define
\begin{align}
  \varepsilon_{\rho,{\rm v}}
  &:=\min\!\left\{1,
  \Lambda_{\mathsf M}
  \sqrt{\frac{\log(4n_{\mathrm v}T|\mathcal Y|/\delta_{\mathrm{tom}})}{2N_{\mathrm v}}}
  \right\},\\
  \varepsilon_{\rm tom}
  &:=\min\!\left\{2\log(1/\gamma),
  \frac{4\sqrt{2\varepsilon_{\rho,{\rm v}}}}{\gamma}
  \right\},\\
  \varepsilon_{\rm samp}
  &:=2\log(1/\gamma)\sqrt{\frac{\log(T/\delta_{\mathrm v})}{2n_{\mathrm v}}}.
  \label{eq:finite-shot-validation-slacks}
\end{align}
Then, with probability at least $1-\delta_{\mathrm v}-\delta_{\mathrm{tom}}$, simultaneously for all $t$,
\begin{equation}
  L_{\gamma,t}(\widehat\theta)
  \leq \widetilde L_{\gamma,t}(\widehat\theta)
  +\varepsilon_{\rm samp}+\varepsilon_{\rm tom}.
  \label{eq:finite-shot-validation-bound}
\end{equation}
If
\begin{equation}
  \max_t\pos{\widetilde L_{\gamma,t}(\widehat\theta)-\widehat\Delta_t}
  \leq\eta_{\mathrm v},
\end{equation}
then, on the additional pilot event in \cref{thm:finite-shot-clock}, let $\varepsilon_{\rm clk}$ be any simultaneous upper bound on $\max_t|\widehat\Delta_t-\Delta_t^{\rm true}|$; under exact inversion one may take $\varepsilon_{\rm clk}=4\mathfrak f_d(\varepsilon_\rho)$, while \cref{thm:inexact-clock} adds $2\varepsilon_{\rm inv}$ for inexact inversion.  Then
\begin{equation}
  \max_t\pos{L_{\gamma,t}(\widehat\theta)-\Delta_t^{\mathrm{true}}}
  \leq
  \eta_{\mathrm v}+\varepsilon_{\rm samp}+\varepsilon_{\rm tom}+\varepsilon_{\rm clk}.
  \label{eq:finite-shot-population-excess}
\end{equation}
The joint confidence is at least $1-\delta_\chi-\delta_{\mathrm v}-\delta_{\mathrm{tom}}$.
\end{theorem}

\begin{corollary}[End-to-end finite-resource endpoint bound]
\label{cor:finite-resource-endpoint}
Under \cref{thm:finite-shot-clock,thm:finite-shot-validation}, set
\begin{equation}
  \varepsilon_{\rm FR}
  :=\eta_{\mathrm v}+\varepsilon_{\rm samp}
  +\varepsilon_{\rm tom}+\varepsilon_{\rm clk}.
\end{equation}
With joint probability at least $1-\delta_\chi-\delta_{\mathrm v}-\delta_{\mathrm{tom}}$,
\begin{equation}
  D_0
  \leq
  \left(\prod_{j=1}^T\kappa_j\right)D_T
  +\sum_{t=1}^T
  \left(\prod_{j=1}^{t-1}\kappa_j\right)
  \psi_\gamma(\Delta_t^{\mathrm{true}}+\varepsilon_{\rm FR}).
  \label{eq:finite-resource-endpoint}
\end{equation}
For input-independent latent-conditioned CPTP kernels, $\kappa_t\leq1$, giving the corresponding worst-step bound with $\varepsilon_{\rm FR}$ in place of $\varepsilon_{\rm cal}$.
\end{corollary}

\subsection{Proof of the finite-shot clock theorem}
For fixed $x,y$, Hoeffding's inequality gives
\begin{equation}
  \Pr\!\left(\abs{\widehat q_{xy}-q_{xy}}>a\right)
  \leq2e^{-2N_\chi a^2}.
\end{equation}
A union bound over $m|\mathcal Y|$ pairs, with
\begin{equation}
  a=\sqrt{\frac{\log(2m|\mathcal Y|/\delta_\chi)}{2N_\chi}},
\end{equation}
shows that all frequency errors are at most $a$ with probability $1-\delta_\chi$. On this event,
\begin{equation}
  \norm{\widetilde\rho_x-\rho_x}_1
  \leq\sum_{y\in\mathcal Y}\abs{\widehat q_{xy}-q_{xy}}\norm{F_y}_1
  \leq\Lambda_{\mathsf M}a.
\end{equation}
Because $\widehat\rho_x$ is a trace-norm projection onto the compact state space and the true state is feasible,
\begin{align}
  \norm{\widehat\rho_x-\rho_x}_1
  &\leq\norm{\widehat\rho_x-\widetilde\rho_x}_1
  +\norm{\widetilde\rho_x-\rho_x}_1\\
  &\leq2\norm{\widetilde\rho_x-\rho_x}_1.
\end{align}
Thus $\dtr(\widehat\rho_x,\rho_x)\leq\Lambda_{\mathsf M}a$, capped by one. Contractivity of trace distance under $\cD_\lambda$ gives the same bound for all $\lambda$. Convexity of the trace norm gives
\begin{equation}
  \dtr\!\left(
  \sum_xp_x\widehat\rho_x(\lambda),
  \sum_xp_x\rho_x(\lambda)
  \right)
  \leq\sum_xp_x\dtr(\widehat\rho_x(\lambda),\rho_x(\lambda))
  \leq\varepsilon_\rho.
\end{equation}
The sharp entropy-continuity inequality of \citet{audenaert2007sharp}, maximized over trace distances no larger than $\varepsilon_\rho$, gives
\begin{equation}
  \abs{S(\widehat\rho)-S(\rho)}\leq\mathfrak f_d(\varepsilon_\rho).
\end{equation}
Applying it once to the average state and once to every ensemble member yields \cref{eq:uniform-chi-shot-bound}. Finally, $|\widehat\Delta_{\rm tot}-\Delta_{\rm tot}|\leq4\mathfrak f_d(\varepsilon_\rho)$, so $\Delta_{\rm tot}>4\mathfrak f_d(\varepsilon_\rho)$ implies $\widehat\Delta_{\rm tot}>0$. The forward-clock theorem applied to the reconstructed nontrivial ensemble then gives a unique estimated grid, and \cref{prop:clock-estimation} with $\varepsilon_\chi=2\mathfrak f_d(\varepsilon_\rho)$ gives \cref{eq:clock-shot-exact-app}.
\hfill$\square$

\subsection{Proof of the finite-shot validation theorem}
There are at most $2n_{\mathrm v}T$ states to reconstruct. Repeating the preceding Hoeffding--projection argument with a union bound over $2n_{\mathrm v}T|\mathcal Y|$ state--outcome pairs gives, with probability at least $1-\delta_{\mathrm{tom}}$, trace-distance error at most $\varepsilon_{\rho,{\rm v}}$ for every reconstructed state. The fidelity continuity estimate \cref{eq:fidelity-continuity}, applied to each argument in turn, gives
\begin{equation}
  \abs{\Fr(\rho,\sigma)-\Fr(\widehat\rho,\widehat\sigma)}
  \leq2\sqrt{2\varepsilon_{\rho,{\rm v}}}.
\end{equation}
Because $f\mapsto-2\log\max\{\gamma,f\}$ is $2/\gamma$-Lipschitz and has range $[0,2\log(1/\gamma)]$, every tomographic loss differs from its ideal counterpart by at most $\varepsilon_{\rm tom}$. Conditional on the fixed trained model, Hoeffding's inequality and a union bound over timesteps give
\begin{equation}
  L_{\gamma,t}(\widehat\theta)
  \leq\widehat L_{\gamma,t}^{\mathrm{ideal}}(\widehat\theta)
  +\varepsilon_{\rm samp}
\end{equation}
simultaneously for all $t$, with probability $1-\delta_{\mathrm v}$. Combining the two events gives \cref{eq:finite-shot-validation-bound}. Adding the pilot budget error and taking positive parts proves \cref{eq:finite-shot-population-excess}.
\hfill$\square$

\subsection{Proof of the finite-resource endpoint corollary}
On the joint pilot and validation event, \cref{eq:finite-shot-population-excess} gives
\begin{equation}
  L_{\gamma,t}(\widehat\theta)\leq\Delta_t^{\mathrm{true}}+\varepsilon_{\rm FR}
\end{equation}
for every $t$. The proof of \cref{thm:endpoint-certificate} then applies verbatim.
\hfill$\square$

\section{Restricted-channel and Choi certificates}
\label{app:architecture}
Universal recoverability is an unrestricted-channel statement. This section quantifies the loss incurred by a restricted Stinespring ansatz and gives a post-training channel certificate.

For a stochastic model, the corresponding barycentric (average-output) channel is
\begin{equation}
  \overline{\cR}_{\theta,t}(\rho)
  :=\int_{\mathcal Z}\cR_{\theta,t,z}(\rho)\pi_t(\dd z).
  \label{eq:barycentric-channel-main}
\end{equation}
It is CPTP but does not determine the full state-valued law $K_{\theta,t}(\rho,\cdot)$.

\begin{theorem}[Restricted-class attainability from Stinespring approximation]
\label{thm:stinespring-approximation}
Assume $V_{\theta,t}$ and $V_t^{\mathrm U}$ are Stinespring isometries into a common dilation space $Q\otimes E$, and let $\mathsf U(E)$ be the unitary group on the environment. For any $\theta$ and timestep $t$, set
\begin{equation}
  \varepsilon_{{\rm St},t}(\theta):=
  \inf_{W\in\mathsf U(E)}
  \norm{(\id_Q\otimes W)V_{\theta,t}-V_t^{\mathrm U}}_\infty,
  \qquad
  \varepsilon_{\rm app}:=\inf_{\theta\in[-\pi,\pi]^P}\max_{1\le t\le T}\varepsilon_{{\rm St},t}(\theta).
  \label{eq:restricted-approx-radius}
\end{equation}
Then
\begin{equation}
  \half\norm{\cR_{\theta,t}-\cR_t^{\mathrm U}}_\diamond
  \leq \varepsilon_{{\rm St},t}(\theta).
  \label{eq:stinespring-to-diamond}
\end{equation}
Let
\begin{equation}
  \delta^{\rm rel}_{x,t}:=D(\rho_{x,t-1}\Vert\bar\rho_{t-1})
  -D(\rho_{x,t}\Vert\bar\rho_t),
  \qquad \Delta_t=\sum_xp_x\delta^{\rm rel}_{x,t}.
\end{equation}
For a deterministic model, interpret $L_{\gamma,t}$ using a degenerate latent law concentrated on the single branch. Then
\begin{align}
  \Fr(\rho_{x,t-1},\cR_{\theta,t}(\rho_{x,t}))
  &\geq\pos{e^{-\delta^{\rm rel}_{x,t}/2}-\sqrt{2\varepsilon_{{\rm St},t}(\theta)}},
  \label{eq:individual-approx-fidelity}\\
  L_{\gamma,t}(\theta)
  &\leq \mathcal A_{\gamma,t}(\varepsilon_{{\rm St},t}(\theta)),
  \label{eq:restricted-class-risk}
\end{align}
where
\begin{equation}
  \mathcal A_{\gamma,t}(\varepsilon)
  :=\sum_xp_x\left[-2\log\max\!\left\{
  \gamma,\pos{e^{-\delta^{\rm rel}_{x,t}/2}-\sqrt{2\varepsilon}}
  \right\}\right].
  \label{eq:A-approximation}
\end{equation}
If the infimum defining $\varepsilon_{\rm app}$ in \cref{eq:restricted-approx-radius} is attained, a minimizing shared parameter satisfies $\varepsilon_{{\rm St},t}(\theta)\leq\varepsilon_{\rm app}$ for every $t$. In general, for every $\upsilon>0$ there exists $\theta_\upsilon$ with $\max_t \varepsilon_{{\rm St},t}(\theta_\upsilon)\leq\varepsilon_{\rm app}+\upsilon$. At $\varepsilon=0$, $\mathcal A_{\gamma,t}(0)\leq\Delta_t$.
\end{theorem}

For a single timestep, an unrestricted Stinespring unitary with ancilla dimension at least $d^2$ can realize $\cR_t^{\mathrm U}$ exactly. A shared conditional model has zero approximation radius whenever it realizes the finite family $\{V_t^{\mathrm U}\}_{t=1}^T$ at the corresponding conditions.

For a linear map $\Phi:Q\to Q$, $\norm{\Phi}_\diamond$ denotes the standard completely bounded trace norm. Define its normalized Choi operator
\begin{equation}
  J(\Phi):=(\idmap_R\otimes\Phi)(\ketbra{\Omega}),
  \qquad
  \ket{\Omega}:=\frac1{\sqrt d}\sum_{i=1}^d\ket{i}^{R}\ket{i}^{Q}.
  \label{eq:normalized-choi}
\end{equation}
For a channel, $J(\Phi)$ is a density operator; this is the normalized Choi--Jamio{\l}kowski representation \citep{choi1975completely,watrous2009sdp}.

\begin{theorem}[Finite-shot Choi certificate for teacher approximation]
\label{thm:choi-validation}
Let $\widetilde\cR_t^{\mathrm U}$ be a numerically implemented teacher satisfying
\begin{equation}
  \dtr(J(\widetilde\cR_t^{\mathrm U}),J(\cR_t^{\mathrm U}))
  \leq\varepsilon_{{\rm teach},t}.
\end{equation}
After training a deterministic channel $\cR_{\widehat\theta,t}$, perform tomography on $N_J$ copies of the learned normalized Choi state at each timestep using an informationally complete POVM $\mathsf M_J=\{M_{J,y}\}_{y\in\mathcal Y_{\rm Choi}}$ on dimension $d^2$, with dual-frame condition number $\Lambda_J$. Let $\widehat J_t$ be the physical estimate and define
\begin{equation}
  \varepsilon_J:=\min\!\left\{1,
  \Lambda_J\sqrt{\frac{\log(2T|\mathcal Y_{\rm Choi}|/\delta_J)}{2N_J}}
  \right\}.
\end{equation}
Then, with probability at least $1-\delta_J$, simultaneously for all $t$,
\begin{equation}
  \half\norm{\cR_{\widehat\theta,t}-\cR_t^{\mathrm U}}_\diamond
  \leq\varepsilon_{\diamond,t}^{\rm Choi},
  \label{eq:choi-diamond-certificate}
\end{equation}
where
\begin{equation}
  \varepsilon_{\diamond,t}^{\rm Choi}
  :=\min\!\left\{1,
  d\left[
  \dtr(\widehat J_t,J(\widetilde\cR_t^{\mathrm U}))
  +\varepsilon_J+\varepsilon_{{\rm teach},t}
  \right]\right\}.
  \label{eq:choi-diamond-error}
\end{equation}
Substituting $\varepsilon_{\diamond,t}^{\rm Choi}$ into \cref{prop:approximation} gives a post-training average-fidelity certificate, and the same half-diamond control yields a clipped-risk certificate through \cref{eq:A-approximation}. For a stochastic model, applying the construction to the barycentric channel $\overline{\cR}_{\widehat\theta,t}$ certifies the averaged channel, while branchwise expected loss is validated directly by \cref{thm:finite-shot-validation}.
\end{theorem}

\subsection{Proof of the Stinespring-approximation theorem}
Environment unitaries do not change the induced channel, so fix one attaining the infimum up to an arbitrarily small error and write the aligned isometries as $V$ and $W$. Because the difference of two channels is Hermiticity preserving, its diamond norm can be optimized on a pure input with a reference system of input dimension. For such an input, the trace distance between the two pure dilated outputs is at most the Euclidean distance of their state vectors, hence at most $\norm{V-W}_\infty$. Partial trace cannot increase trace distance. Taking the supremum gives
\begin{equation}
  \half\norm{\cR_{\theta,t}-\cR_t^{\mathrm U}}_\diamond
  \leq\norm{V-W}_\infty,
\end{equation}
consistent with the continuity theorem for Stinespring representations \citep{kretschmann2008continuity}. Applying \cref{eq:universal-recovery} to the individual pair indexed by $x$ gives
\begin{equation}
  \Fr(\rho_{x,t-1},\cR_t^{\mathrm U}(\rho_{x,t}))
  \geq e^{-\delta^{\rm rel}_{x,t}/2}.
\end{equation}
The fidelity-continuity argument in \cref{prop:approximation} then proves \cref{eq:individual-approx-fidelity}. Applying the clipped loss, averaging over $x$, and using its monotonic decrease in fidelity gives \cref{eq:restricted-class-risk}. At $\varepsilon=0$, each summand equals $\min\{\delta^{\rm rel}_{x,t},2\log(1/\gamma)\}$, whose average is at most $\Delta_t$.
\hfill$\square$

\subsection{A normalized-Choi norm lemma}
For a Hermiticity-preserving map $\Phi:Q\to Q$ and the normalized Choi operator in \cref{eq:normalized-choi},
\begin{equation}
  \norm{J(\Phi)}_1
  \leq\norm{\Phi}_\diamond
  \leq d\norm{J(\Phi)}_1.
  \label{eq:choi-diamond-norm-equivalence}
\end{equation}
The lower bound follows by evaluating the diamond norm on $\ketbra{\Omega}$. For the upper bound, any normalized pure reference-assisted input can be written as $(A\otimes\id)\ket{\Omega}$ with $\norm{A}_2^2=d$. Therefore
\begin{align}
  \norm{(\idmap\otimes\Phi)(\ketbra{\psi})}_1
  &=\norm{(A\otimes\id)J(\Phi)(A^\dagger\otimes\id)}_1\\
  &\leq\norm{A}_\infty^2\norm{J(\Phi)}_1
  \leq d\norm{J(\Phi)}_1.
\end{align}
For Hermiticity-preserving maps the diamond norm may be optimized on a pure input with a reference of input dimension, yielding the upper bound; see also \citet{watrous2009sdp}.

\subsection{Proof of the Choi-validation theorem}
The IC-tomography argument of \cref{thm:finite-shot-clock}, now in dimension $d^2$ and union-bounded over $T$ Choi states, gives
\begin{equation}
  \dtr(\widehat J_t,J(\cR_{\widehat\theta,t}))\leq\varepsilon_J
\end{equation}
for all $t$ with probability $1-\delta_J$. The triangle inequality then gives
\begin{align}
  \dtr(J(\cR_{\widehat\theta,t}),J(\cR_t^{\mathrm U}))
  \leq{}&\dtr(\widehat J_t,J(\widetilde\cR_t^{\mathrm U}))
  +\varepsilon_J+\varepsilon_{{\rm teach},t}.
\end{align}
Applying the upper bound in \cref{eq:choi-diamond-norm-equivalence} to the difference of the channels proves \cref{eq:choi-diamond-certificate}; the cap at one uses the fact that half the diamond distance between channels is at most one.
\hfill$\square$

\clearpage
\section{Additional validation diagnostics}
\label{app:validation-diagnostics}

This section collects additional quantitative checks supporting the main-text calibration results: a four-qubit generation view, a monotone-proxy identifiability audit, the numerical common-channel SDP audit, matched local trace-error intervention effects, a frozen broader-field stress test, and an endpoint-certificate audit. These checks reuse the frozen runs and saved evaluations rather than introducing new model selection.

\subsection{Additional four-qubit generation view}
\begin{figure}[htbp]
  \centering
  \includegraphics[width=.62\linewidth]{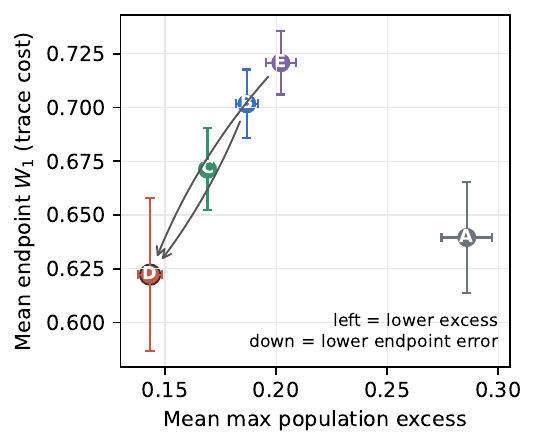}
  \caption{\textbf{Primary four-qubit calibration--generation comparison.} Lower maximum excess is better to the left and lower endpoint $\Wtr$ is better downward. D has the lowest mean excess and endpoint error among the five variants.}
  \label{fig:fourq-pareto}
\end{figure}

\label{app:additional-validation}

\subsection{Monotone-proxy identifiability audit}
\label{app:proxy-audit}
For transparency, we compare the A/E-only within-run Spearman alignment of $\Delta_t$ with several simple monotone coordinates: timestep $t$, cumulative corruption $1-\lambda_t$, one-step depolarizing strength $q_t=1-\lambda_t/\lambda_{t-1}$, and $\Delta\lambda_t=\lambda_{t-1}-\lambda_t$. With prespecified signed orientations, the first three proxies are strongly negatively correlated with reverse loss whereas $\Delta_t$ is positively correlated; this sign contrast is not interpreted as predictor superiority because orientation is arbitrary. More importantly, with only $T=8$ monotone steps, ranked $\Delta_t$ is collinear with ranked $t$ or $q_t$ in all $60/60$ controlled two-qubit A/E runs, $17/20$ primary four-qubit A/E runs, and $18/20$ broader-field A/E runs. The corresponding partial-rank effect is therefore undefined in most runs, and the few estimable values are not positive. We consequently make no claim that the observed Spearman alignment identifies incremental rank information beyond timestep or noise strength. Its role is the more limited mechanism check that the information decrement orders the realized reverse-step difficulty along the tested paths.

\subsection{Numerical common-channel SDP audit}
For the controlled two-qubit ensembles, deterministic de-duplication leaves nine unique target/grid objects (three difficulty levels times linear, equal-cq-MI, and cosine grids). For each of their eight timesteps we solve
\begin{equation}
 \min_{\cR\,\mathrm{CPTP}}\;\sum_xp_x\dtr(\rho_{x,t-1},\cR(\rho_{x,t}))
 \label{eq:numerical-recovery-sdp}
\end{equation}
using a Choi-matrix semidefinite program. The primary SCS solves use tolerance $10^{-6}$ and are cross-checked with Clarabel at one representative timestep for each target/grid object; the maximum cross-solver objective discrepancy is $7.15\times10^{-7}$. All $72$ timestep optima satisfy the continuity-based component of the intrinsic-recovery bracket at declared tolerance $7\times10^{-5}$. The minimum/median/maximum gaps above the Holevo-continuity lower bound are $0.00970/0.04329/0.18336$, while the corresponding gaps below the universal upper bound are $0.00822/0.16294/0.36334$.

\begin{figure}[htbp]
  \centering
  \includegraphics[width=.80\linewidth]{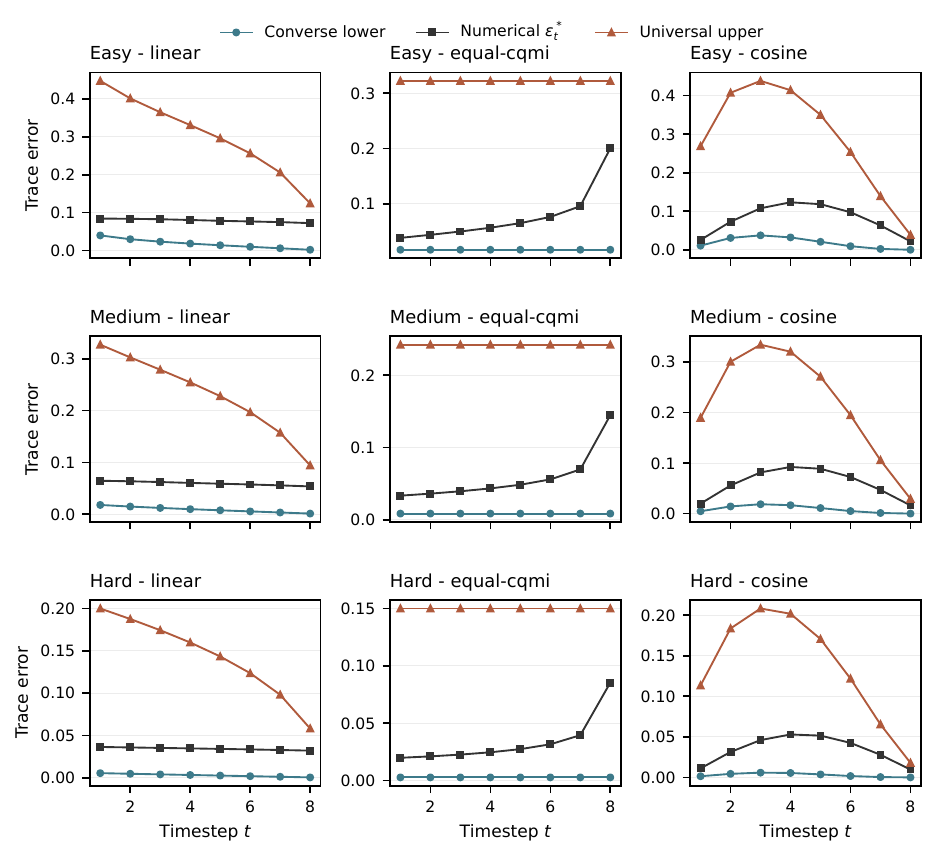}
  \caption{\textbf{Numerical two-qubit continuity-recovery sandwich.} The Holevo-continuity lower bound, numerical common-channel optimum, and universal upper bound are shown for nine de-duplicated target/grid objects.}
  \label{fig:twoq-sandwich}
\end{figure}

\subsection{Independent local trace-error effects}
The paired intervention effects below use constrained minus matched unconstrained runs; negative values denote improvement. Maximum local trace error is $\max_t\max_{x,z}\dtr(\rho_{x,t-1},\cR_{t,z}(\rho_{x,t}))$, while mean local trace error averages first over the finite $(x,z)$ support and then over timesteps.
\begin{table}[htbp]
\centering
\caption{\textbf{Constraint intervention on a non-optimized local trace metric.} Entries are paired mean changes with 95\% matched-seed bootstrap intervals; negative values denote improvement.}
\label{tab:local-trace-effects-appendix}
\resizebox{\linewidth}{!}{%
\begin{tabular}{llrr}
\toprule
Study & Comparison & Max local trace change & Mean local trace change\\
\midrule
Two-qubit & C--A & $-0.00234\;[-0.00327,-0.00148]$ & $-0.00088\;[-0.00097,-0.00080]$\\
Two-qubit & D--B & $-0.04185\;[-0.04855,-0.03479]$ & $-0.04892\;[-0.05891,-0.03895]$\\
Primary 4q & C--A & $-0.21988\;[-0.25443,-0.18205]$ & $-0.03021\;[-0.03864,-0.02070]$\\
Primary 4q & D--B & $-0.11896\;[-0.15414,-0.08746]$ & $-0.04274\;[-0.05259,-0.03363]$\\
Broader 4q & C--A & $-0.20044\;[-0.23590,-0.16485]$ & $-0.02684\;[-0.03712,-0.01734]$\\
Broader 4q & D--B & $-0.17546\;[-0.22145,-0.12896]$ & $-0.04644\;[-0.05913,-0.03337]$\\
\bottomrule
\end{tabular}}
\end{table}
All twelve displayed intervals exclude zero. For maximum local trace error, the improvement counts are $25/30$ for two-qubit C--A, $30/30$ for two-qubit D--B, and $10/10$ for each four-qubit comparison. Thus the constraint intervention improves a separately evaluated local quantity across every retained regime, including the frozen broader-field shift.

\begin{figure}[htbp]
  \centering
  \includegraphics[width=.80\linewidth]{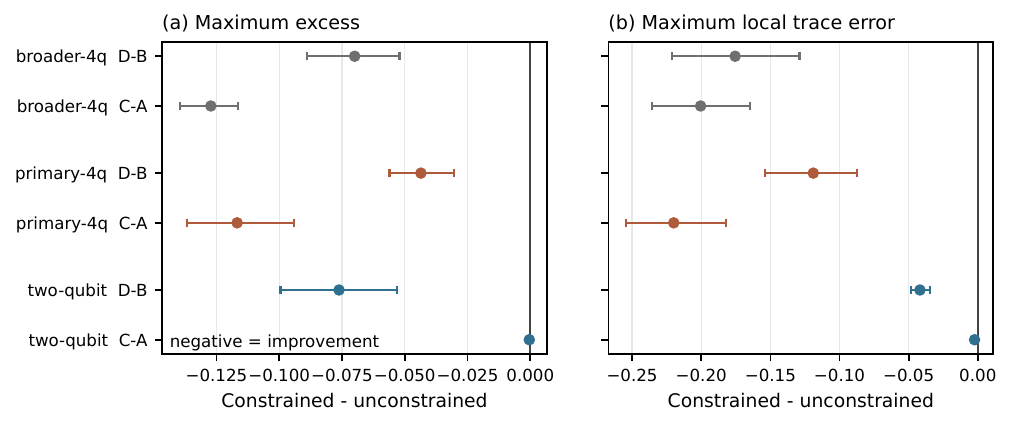}
  \caption{\textbf{Constraint effects on the defining residual and a non-optimized local outcome.} Negative values indicate improvement.}
  \label{fig:constraint-local-trace}
\end{figure}

\subsection{Frozen broader-field four-qubit stress test}
\label{app:broader-field-results}
The broader-field experiment freezes the primary four-qubit architecture, optimization, data split, and A--E definitions while changing only the TFIM field range from $g\sim U[0.2,0.4]$ to $g\sim U[0.6,1.4]$. The target pairwise trace-distance diversity is $2.01\times$ larger and the target observable spread is $4.94\times$ larger than in the primary regime. Despite this shift, A/E information--loss alignment remains strong ($\rho=0.880$, 95\% CI $[0.824,0.933]$, $20/20$ positive). Constraint enforcement also remains effective: C--A changes maximum excess by $-0.1272$ (CI $[-0.1397,-0.1164]$) and D--B by $-0.0699$ (CI $[-0.0895,-0.0520]$), while the maximum local trace-error changes in \cref{tab:local-trace-effects-appendix} are favorable for all ten matched seeds in both comparisons.

The endpoint ordering, however, changes: relative to cosine E, D has larger endpoint $\Wtr$ by $0.0552$ (CI $[0.0097,0.1051]$). This frozen-capacity reversal is useful rather than contradictory: it shows that the information budget and its enforcement remain stable local recoverability diagnostics even when they do not determine the distribution-level ranking. The observation therefore provides an empirical stress test of the recoverability--coverage distinction instead of an additional claim of uniform endpoint superiority.

\subsection{Endpoint-certificate audit}
\label{app:certificate-audit}
We also restore the finite-ensemble endpoint audit used to check the local-to-global propagation numerically. For each retained run, the \emph{direct} bound substitutes the exactly evaluated finite-support local losses into the one-step trace-Wasserstein propagation, whereas the \emph{budget-calibrated} bound uses the corresponding exact population maximum excess in the worst-step form of Appendix~\ref{app:endpoint}. The audit is a consistency check on matched finite target ensembles; it is not an independent held-out generalization theorem.

\begin{table}[htbp]
\centering
\caption{\textbf{Finite-ensemble endpoint-certificate audit.} No retained run violates either reported bound. Ratios are median bound divided by observed endpoint $\Wtr$; values above one quantify conservativeness.}
\label{tab:endpoint-certificate-audit}
\begin{tabular}{lrrr}
\toprule
Regime & Runs / violations & Direct ratio & Budget ratio\\
\midrule
Controlled 2q & $150/0$ & $3.77$ & $5.68$\\
Primary 4q TFIM & $50/0$ & $3.03$ & $4.85$\\
Broader-field 4q & $50/0$ & $3.28$ & $5.04$\\
\bottomrule
\end{tabular}
\end{table}
Across all $250$ retained runs, the observed endpoint error lies on the certified side of both finite-ensemble bounds. The bounds are therefore numerically valid in these regimes but conservative, typically exceeding the observed error by roughly a factor of three to six. This distinction is important: the theorem provides one-sided compositional control, while the recoverability--coverage separation explains why local-budget propagation need not tightly rank stochastic generation quality.

\clearpage
\section{Additional one-qubit sanity checks}
\label{app:oneq-sanity}
The main text uses the entangled two-qubit study as the controlled mechanism test and the four-qubit TFIM study as the scaling test. The original noncommuting one-qubit experiments provide an additional lower-dimensional sanity check: forward information loss strongly aligns with empirical reverse-step loss and budget constraints reduce the calibrated excess.

\begin{figure}[htbp]
  \centering
  \includegraphics[width=.80\linewidth]{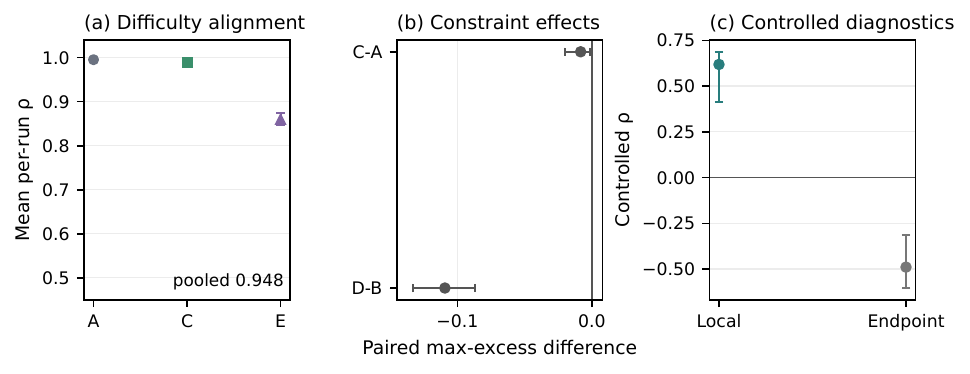}
  \caption{\textbf{Additional one-qubit sanity checks.} The one-qubit subset reproduces the calibration mechanism: strong information--loss alignment and reduced excess under the constraints. These runs are reported separately from the main two-/four-qubit statistics.}
  \label{fig:oneq-sanity}
\end{figure}

\section{Scope and limitations}
\label{app:scope}
The theory is stated for finite-dimensional labeled ensembles and CPTP forward/reverse dynamics. The cleanest scheduling result is for the fixed depolarizing path, where the cq-information clock is invertible; equal information is minimax for the information decrement and the associated worst-step-derived propagation bound, not a claim that it exactly minimizes the unrestricted optimal recovery error for every ensemble geometry. The attainable recovery budget is an unrestricted common-channel reference, while \cref{thm:stinespring-approximation,thm:choi-validation} quantify approximation by restricted Stinespring implementations without guaranteeing successful nonconvex optimization or exact feasibility for a shared ansatz. The uniform bounded-generator result applies directly to the reported four-qubit architecture; the controlled two-qubit study uses a separate CPTP parameterization and evaluates its finite-ensemble calibration quantities exactly.

The finite-sample and tomography guarantees are deliberately worst-case and may be conservative. In particular, the explicit informationally complete tomography route in Appendix~\ref{app:finite-shot} scales poorly with Hilbert-space dimension; structured tomography, classical shadows, or simulator access would be needed for substantially larger systems. Likewise, the empirical ``population'' calibration quantities in the TFIM study are exact expectations over the finite training ensemble and latent support, not guarantees over an unseen continuous field distribution. The endpoint theorem is a one-sided trace-Wasserstein guarantee and need not tightly rank every task-specific endpoint metric, as the broader-field stress test illustrates. Within this scope, the global ensemble-label cq-information coordinate complements local spatial-QCMI recovery constructions by connecting forward information loss, common-channel recoverability, and stochastic generative coverage in a single finite-dimensional framework.

\section{Detailed experimental protocol}
\label{app:experimental-protocol}

\subsection{Controlled two-qubit study}
The controlled study uses entangled two-qubit mixed-state ensembles at three fixed difficulty levels. Each condition contains $8$ equiprobable target states, $T=8$ reverse steps, complete terminal depolarization, and $16$ latent categories. Five variants and ten matched seeds give $3\times5\times10=150$ runs. All variants use the same fixed cq-MI time features $(u_t,h_t)$, latent law, terminal prior, distributional objective, optimizer family, and matched initialization and sampling streams. The executed eight-parameter reverse model is a CPTP rotation-plus-reset density-matrix map: six parameters enter a joint Hamiltonian exponential and two enter sigmoid mixture weights. Its finite-ensemble calibration quantities are evaluated exactly over the label and latent supports.

All five variants use the same $150$-step base phase with primal learning rate $0.025$ and the same fixed $300$-step polish with primal learning rate $0.005$. Constrained variants C and D additionally use one nonnegative multiplier per timestep with dual learning rate $\eta_\alpha=0.20$. All one- and two-qubit runs use the common clipping level $\gamma=10^{-2}$. Endpoint distributions use $2048$ generated trajectories. Population calibration quantities are evaluated exactly over the fixed finite label and latent supports, and condition-controlled diagnostic correlations remove difficulty$\times$variant offsets.

\subsection{Four-qubit TFIM study}
The scaling study uses ground states of the open-boundary four-qubit transverse-field Ising Hamiltonian
\[
H(g)=-\sum_{i=1}^{3}Z_iZ_{i+1}-g\sum_{i=1}^{4}X_i.
\]
For each dataset seed in $\{0,1,\ldots,9\}$, the transverse field is sampled uniformly from $[0.2,0.4]$; the first $100$ generated ground states form the training ensemble and the next $100$ form a held-out ensemble. The executed clock uses these states equiprobably, so its mean-based implementation exactly matches $p_x=1/m$ in the theoretical cq-MI. The forward process is the same cumulative depolarizing family as in the controlled experiments, with $T=8$ and complete terminal depolarization. All variants use $16$ latent categories and the common clipping level $\gamma=10^{-2}$.

Model capacity is selected once on a disjoint pilot seed ($999$), without inspecting held-out endpoint accuracy metrics such as endpoint $\Wtr$ or the TFIM observable. The pre-specified candidates have depths $8,16,24,32$; we choose the smallest depth whose pilot run has maximum population excess at most $0.15$, generated-to-target diversity ratio at least $0.25$, no numerical instability, and non-identical latent outputs. The diversity and non-identical-output gates are coarse non-collapse checks rather than optimization of the reported endpoint accuracy metrics. Depth $8$ is the first candidate satisfying these gates, giving two Stinespring ancillas and $144$ trainable bounded-generator rotations. Each trainable coordinate is an $R_y$ or $R_z$ angle and therefore has an exactly channel-equivalent representative in $[-\pi,\pi]$ by $2\pi$ periodicity; the fixed time/latent conditioning contains no trainable classical parameters. Thus this four-qubit architecture is a literal instance of the bounded-generator class used in the uniform generalization theorem. This architecture is then frozen for A--E and all final seeds. The depth-$128$ result reported in \cref{sec:exp-capacity-external} is a separate capacity extension of D, not a replacement for this frozen A--E architecture; it changes only the repeated Stinespring depth and is evaluated on the same ten primary seeds without using held-out endpoint metrics for stability selection.

\begin{table}[htbp]
\centering
\caption{\textbf{Pre-specified four-qubit capacity-selection rule.} Architecture selection uses the disjoint pilot seed only and never uses held-out endpoint $\Wtr$ or the TFIM observable; the diversity and latent-output gates are coarse non-collapse checks. The smallest candidate passing every gate is frozen for all A--E runs.}
\label{tab:capacity-selection}
\begin{tabular}{ll}
\toprule
Pilot item & Pre-specified value / outcome\\
\midrule
Candidate depths & $\{8,16,24,32\}$\\
Selection seed & $999$ (excluded from final statistics)\\
Maximum-excess gate & $\leq 0.15$\\
Diversity-ratio gate & $\geq 0.25$\\
Additional gates & finite numerics; non-identical latent outputs\\
Selected architecture & depth $8$, $144$ trainable rotations\\
Selection uses held-out $\Wtr$/observable? & no\\
\bottomrule
\end{tabular}
\end{table} Every final run uses $600$ base Adam steps at learning rate $0.01$ followed by $400$ polish steps at learning rate $0.002$; constrained variants use dual learning rate $0.20$ throughout. The effective batch size is $16$ (microbatch $16$ with accumulation factor $1$). Endpoint evaluation uses $1024$ generated samples. The reported model is the final polish iterate, and the held-out set is reserved for the reported generation metrics rather than model-capacity selection, training, or checkpointing. Calibration population quantities are exact expectations over the $100$-state training ensemble and the exhaustive $16$-point latent support. Calibration and endpoint diagnostics use matched finite target ensembles for the decrement, local recovery loss, and observed endpoint distance. The recorded median total runtime is $160$ seconds per run, with peak RAM fraction below $0.20$ on the laptop environment used for the study.

For endpoint evaluation, the TFIM table reports a linear Hilbert--Schmidt MMD$^2$, distinct from the Gaussian-mixture MMD$^2$ used for training. With $k_{\rm HS}(\rho,\sigma)=\operatorname{Re}\Tr(\rho\sigma)$, the reported estimator is the unbiased U-statistic
\[
\widehat{\mathrm{MMD}}_{\rm HS,u}^{2}(A,B)=
\frac{\sum_{i\ne i'}k_{\rm HS}(a_i,a_{i'})}{n(n-1)}+
\frac{\sum_{j\ne j'}k_{\rm HS}(b_j,b_{j'})}{m(m-1)}-
\frac{2}{nm}\sum_{i,j}k_{\rm HS}(a_i,b_j),
\]
computed between generated endpoint states and held-out target states. The ``Max local error'' column is instead the training-ensemble quantity
\[
E_{\rm local}^{\max}(\theta)=\max_{1\le t\le T}\max_{x,z}
\dtr\!\left(\rho_{x,t-1},\cR_{\theta,t,z}(\rho_{x,t})\right),
\]
with no averaging over $x$ or $z$ before the maxima. The TFIM observable is the computational-basis absolute longitudinal magnetization
\[
M_{|Z|}=\sum_{b\in\{0,1\}^{4}}
\left|\frac14\sum_{i=1}^{4}z_i(b)\right|\ketbra{b},
\qquad z_i(b)\in\{-1,+1\},
\]
and its run-level error is
\[
E_{\rm obs}=\left|\frac1{N_g}\sum_{j=1}^{N_g}\Tr(\widehat\rho_jM_{|Z|})-
\frac1{N_h}\sum_{k=1}^{N_h}\Tr(\rho_k^{\rm heldout}M_{|Z|})\right|.
\]
We additionally report mean pairwise trace-distance diversity and diversity relative to the held-out target ensemble. Pairwise method comparisons use matched dataset seeds. For the cross-scale diagnostic, the four-qubit controlled association removes variant means; the two-qubit analysis removes difficulty$\times$variant means.

\subsection{Capacity extension and external QuDDPM benchmark}
\label{app:capacity-external}
The capacity experiment changes only the shared Stinespring depth from $8$ to $128$. The resulting D-depth128 model remains a single shared conditional reverse model across all eight timesteps and all $16$ latent branches, with two ancillas and $2304$ trainable rotations. It retains the canonical equal-cq-MI schedule and budgets, complete-depolarization endpoints, distributional objective, primal--dual constraints, $\gamma=10^{-2}$, maximally mixed terminal prior, $600+400$ Adam updates at learning rates $0.01/0.002$, dual learning rate $0.20$, final-iterate checkpoint rule, and $1024$ held-out endpoint samples. A seed-$999$ smoke test used no endpoint metric and required no stability correction before the ten primary seeds were executed.

Against canonical depth $8$, D-depth128 lowers mean endpoint $\Wtr$ from $.622278$ to $.424375$ (31.80\%), with all $10/10$ matched seeds favorable and paired bootstrap interval $[-.262345,-.128067]$. Its mean maximum local trace error falls from $.328600$ to $.287385$, while mean maximum population excess falls from $.143233$ to $.092243$. HS-MMD$^2$ and the TFIM observable error also improve relative to depth $8$. These simultaneous local and endpoint gains show that additional shared-channel expressivity strengthens both sides of the learned recovery--coverage decomposition.

For the external comparison, we use the audited official QuDDPM implementation \citep{zhang2024generative} on the same primary four-qubit TFIM arrays, seeds $0,\ldots,9$, $100/100$ training/held-out split, and the same $1024$-sample held-out evaluator. QuDDPM retains its native $30$-step circuit diffusion, two ancillas, $12$ layers, natural-distance loss, Haar terminal input, and $4320$ trainable parameters. D-depth128 achieves lower endpoint $\Wtr$ on every matched seed; the mean values are $.424375$ and $.498060$, with paired D-depth128 minus QuDDPM bootstrap interval $[-.083112,-.064411]$.

\begin{table}[htbp]
\centering
\caption{\textbf{Capacity and external benchmark.} Means over ten matched primary TFIM seeds. Diversity ratio is generated/target pairwise trace-distance diversity; values closer to one indicate closer spread.}
\label{tab:capacity-external-app}
\resizebox{\linewidth}{!}{\begin{tabular}{lrrrrrr}
\toprule
Method & Params. & Endpoint $\Wtr$ & HS-MMD$^2$ & Obs. error & Div. ratio & Runtime (s)\\
\midrule
D-depth8 & $144$ & $.6223$ & $.4196$ & $.3847$ & $2.390$ & $116.2$\\
D-depth128 & $2304$ & $\mathbf{.4244}$ & $.1853$ & $.2609$ & $2.458$ & $2017.8$\\
QuDDPM & $4320$ & $.4981$ & $\mathbf{.0865}$ & $\mathbf{.1701}$ & $10.399$ & $242.3$\\
\bottomrule
\end{tabular}}
\end{table}
D-depth128 uses 53.3\% of QuDDPM's trainable parameter count and achieves 14.8\% lower mean endpoint $\Wtr$. The comparison is metric-complete: QuDDPM has lower HS-MMD$^2$, observable error, and runtime, while D-depth128 has lower endpoint $\Wtr$ and substantially fewer trainable parameters. This distinction is aligned with our theory, which treats trace-Wasserstein distance as the endpoint geometry propagated by the stochastic reverse chain.

As a discretization control, keeping the compact $144$-parameter shared model fixed while increasing $T$ from $8$ to $30$ changes mean endpoint $\Wtr$ only from $.622278$ to $.624896$, even though local budget satisfaction improves. In contrast, increasing shared-model capacity at fixed $T=8$ yields the 31.8\% endpoint reduction above. This isolates reverse-model expressivity, rather than a finer diffusion grid, as the effective scaling axis in this experiment.

\subsection{Broader-field four-qubit TFIM stress test}
The stress test uses the same dataset seeds $0,\ldots,9$ and samples the transverse field from $g\sim U[0.6,1.4]$; this field range is the only scientific setting changed from the primary four-qubit study. It freezes the same Hamiltonian, $100/100$ training/held-out split, $T=8$ complete-depolarization path, $16$ latent categories, $\gamma=10^{-2}$, depth-$8$ two-ancilla Stinespring architecture with $144$ trainable parameters, effective batch size $16$, $600$ Adam steps at learning rate $0.01$ followed by $400$ polish steps at $0.002$, and dual learning rate $0.20$ for C/D. Each run uses $1024$ generated endpoint trajectories; calibration population quantities use the $100$ training states and exhaustive latent support, whereas endpoint generation metrics use the $100$ held-out states. The reported checkpoint is the final polish iterate, with no endpoint or observable metric used for selection. Five variants and ten matched dataset seeds give $50$ runs. No architecture or optimization hyperparameter is retuned for this regime. Relative to the primary TFIM ensemble, the target pairwise trace-distance diversity is $2.01\times$ larger and the target observable spread is $4.94\times$ larger, making this a frozen-capacity distribution-shift stress test.

\subsection{Common variants and primal--dual implementation}
The five variants are:
\begin{center}
\begin{tabular}{lll}
\toprule
Variant & Forward grid & Reverse training rule\\
\midrule
A & linear-in-retention & unconstrained distributional objective\\
B & equal-cq-MI & unconstrained distributional objective\\
C & linear-in-retention & distributional objective + cq-MI constraints\\
D & equal-cq-MI & distributional objective + cq-MI constraints\\
E & cosine-in-retention & unconstrained distributional objective\\
\bottomrule
\end{tabular}
\end{center}
All reported regimes use $\gamma=10^{-2}$. The distributional training objective uses a biased nonnegative Gaussian-mixture MMD$^2$. For state batches $A=\{a_i\}_{i=1}^{n}$ and $B=\{b_j\}_{j=1}^{m}$, define
\[
k(a,b)=\frac13\sum_{\sigma\in\{0.1,0.3,1.0\}}
\exp\!\left[-\frac{\lVert a-b\rVert_F^2}{2\sigma^2}\right],
\]
and
\[
\widehat{\mathrm{MMD}}_{\rm Gmix}^{2}(A,B)=\max\!\left\{0,
\frac1{n^2}\sum_{i,i'}k(a_i,a_{i'})+
\frac1{m^2}\sum_{j,j'}k(b_j,b_{j'})-
\frac{2}{nm}\sum_{i,j}k(a_i,b_j)\right\}.
\]
The distributional objective averages this quantity over the $T$ timesteps. There is no learned bandwidth or median heuristic. For one- and two-qubit endpoint evaluation the reported endpoint MMD uses this same estimator on target and generated endpoint samples; the four-qubit TFIM tables instead use the separate unbiased HS-MMD$^2$ defined above.

For the constrained variants, let $\alpha_t\geq0$ be one multiplier per timestep; we reserve $\lambda_t$ throughout for the forward depolarizing retention level. The empirical Lagrangian and projected dual update are
\begin{align}
  \widehat{\mathcal J}_{\rm PD}(\theta,\bm\alpha)
  &=\widehat{\mathcal J}_{\rm dist}(\theta)
  +\frac1T\sum_{t=1}^T\alpha_t
  \bigl(\widehat L_{\gamma,t}(\theta)-\widehat\Delta_t\bigr),
  \label{eq:primal-dual-objective-main}\\
  \alpha_t
  &\leftarrow
  \pos{\alpha_t+\eta_{\alpha}
  \bigl(\widehat L_{\gamma,t}(\theta)-\widehat\Delta_t\bigr)},
  \qquad \alpha_t\geq0,
  \label{eq:dual-update-main}
\end{align}
with gradient descent in $\theta$ and projected ascent in $\bm\alpha$ \citep{boyd2004convex}; the common factor $1/T$ in the displayed Lagrangian is absorbed into the reported dual stepsize $\eta_\alpha$. The final calibration metric is evaluated separately from the minibatch constraint residual whenever the finite support permits:
\[
L_{\gamma,t}^{\rm pop}
=
\sum_x p_x\sum_z \pi_t(z)\,\ell_{\gamma,t}(\theta;x,z),
\qquad
V_{\max}^{\rm pop}:=\max_t V_{\gamma,t}^{\rm pop}=\max_t\pos{L_{\gamma,t}^{\rm pop}-\Delta_t}.
\]
No held-out endpoint accuracy metric is used for checkpoint selection. For the information-loss/reverse-loss alignment statistic, each run computes $\operatorname{Spearman}_t(\Delta_t,L_{\gamma,t}^{\rm pop})$ across the eight timesteps. The primary statistic pools the unconstrained heterogeneous schedules A and E; constrained C is reported as an auxiliary consistency check. Aggregate correlations are computed from the matched run-level statistics.

The independently seeded schedule comparison reported in
\cref{tab:schedule-confirm-main} uses the same frozen architecture,
optimization rule, clipping level, checkpoint rule, and endpoint metric
as the primary protocol, while datasets and training/evaluation streams
are independently seeded. Only the forward grid is changed across the
linear, equal-cq-MI, and cosine constrained variants.

\subsection{Recoverability--coverage ablation protocol}
\label{app:coverage-ablation-protocol}
The primary four-qubit recoverability--coverage ablation reuses the frozen TFIM setup above and the same ten dataset seeds $0,\ldots,9$. Variants B and D are reused only because their stored configurations, dataset hashes, schedule values, initialization hashes, architecture, optimizer, training length, and endpoint evaluation exactly match the ablation protocol. Variant R uses the same equal-cq-MI grid, stochastic latent support, depth-$8$ Stinespring architecture, initialization stream, $600$ Adam plus $400$ polish steps, and final-iterate checkpoint rule, but removes the distributional MMD objective and minimizes
\[
\widehat{\mathcal J}_{\rm local}(\theta)
=\frac1T\sum_{t=1}^{T}\widehat L_{\gamma,t}(\theta).
\]
No held-out endpoint quantity enters training or model selection, and all ten matched seeds are retained. Because R differs from D only by removing the distributional MMD objective, their nearly identical local trace error and separated endpoint $\Wtr$ directly test the recoverability--coverage decoupling.

\subsection{Additional one-qubit sanity study}
The original noncommuting one-qubit mixed-state study uses the same three difficulty levels, five variants, ten matched seeds per condition, $T=8$, $16$ target states, and $16$ latent categories, giving $150$ additional runs. Its shared circuit has five trainable parameters and follows the same $150$-step base plus $300$-step polish protocol as the controlled two-qubit study. These runs are reported only in Appendix~\ref{app:oneq-sanity} and are excluded from the main two-/four-qubit statistics.

\section{Stability of finite-outcome instruments}
The nonexpansive result in the main text uses an external latent variable. A measurement-conditioned branch can amplify perturbations through normalization, so a general instrument need not be nonexpansive. The following finite-factor estimate records the distinction.

\begin{proposition}[Finite-outcome instrument bound]
Let $\{\mathcal M_y\}_{y=1}^r$ be a quantum instrument, so each $\mathcal M_y$ is completely positive and $\sum_y\mathcal M_y$ is trace-preserving. Let $K_{\mathcal M}(\rho)$ be the probability law of the normalized post-measurement state $\mathcal M_y(\rho)/\Tr\mathcal M_y(\rho)$, with outcome probability $\Tr\mathcal M_y(\rho)$; branches of zero probability may be assigned an arbitrary normalized state because they carry no mass. Then the Dobrushin coefficient in trace-Wasserstein distance satisfies $\kappa(K_{\mathcal M})\leq3$.
\end{proposition}

\begin{proof}
Set $A_y=\mathcal M_y(\rho)$, $B_y=\mathcal M_y(\sigma)$, $p_y=\Tr A_y$, and $q_y=\Tr B_y$. Couple the outcome distributions maximally. The unmatched mass is their total variation distance, which is at most $\dtr(\rho,\sigma)$ by data processing. The matched mass at outcome $y$ is $r_y=\min(p_y,q_y)$. For outcomes with $p_yq_y>0$,
\begin{equation}
  r_y\norm{A_y/p_y-B_y/q_y}_1
  \leq\norm{A_y-B_y}_1+\abs{p_y-q_y}.
\end{equation}
If $p_yq_y=0$, then $r_y=0$ and the matched contribution of that outcome is zero, so no normalized-state quotient is needed. The cq output channel $\rho\mapsto\sum_y\ketbra{y}\otimes\mathcal M_y(\rho)$ is trace-distance contractive, so
\begin{equation}
  \sum_y\norm{A_y-B_y}_1\leq\norm{\rho-\sigma}_1,
  \qquad
  \sum_y\abs{p_y-q_y}\leq\norm{\rho-\sigma}_1.
\end{equation}
The matched contribution is therefore at most $2\dtr(\rho,\sigma)$, and the unmatched contribution is at most $\dtr(\rho,\sigma)$ because the diameter of state space under trace distance is one. Summing gives the factor three.
\end{proof}

\end{document}